\documentclass[12pt]{article}

\usepackage[latin9]{inputenc}
\usepackage{mathtools}
\usepackage{amsmath}
\usepackage{amsthm}
\usepackage{amssymb}
\usepackage{enumitem}
\usepackage{xcolor}
\usepackage{booktabs}
\usepackage{threeparttable}
\usepackage{booktabs}
\usepackage{multirow}
\usepackage{tikz}
\usepackage{float}
\usetikzlibrary{arrows.meta}
\usepackage{graphicx}
\usepackage{hyperref}
\usepackage[a4paper]{geometry}
\usepackage{setspace}
\usepackage[authoryear]{natbib}
\makeatletter

\providecommand{\tabularnewline}{\\}

\theoremstyle{plain}
\newtheorem{prop}{\protect\propositionname}
\newtheorem{thm}{Theorem}

\newcommand{\reductionarrow}[1]{%
  \tikz[baseline=(current bounding box.center),x=1em,y=1ex]{
    \draw[-{Stealth[length=2pt,width=2pt]},gray!75,line width=.35pt]
      (0,1.45) to[out=-65,in=65] (0,-1.45);
    \node[font=\scriptsize,text=gray!80,anchor=west] at (.35,0) {#1};
  }%
}

\providecommand{\propositionname}{Proposition}

\renewcommand{\hat}{\widehat}
\date{\today}

\makeatother

\providecommand{\propositionname}{Proposition}

\begin{document}
\title{Best Feasible Conditional Critical Values for a More Powerful Subvector
Anderson-Rubin Test}
\author{Jesse Hoekstra\thanks{jesse.hoekstra@nuffield.ox.ac.uk} ~and Frank
Windmeijer\thanks{frank.windmeijer@stats.ox.ac.uk}\\
 {\small Dept of Statistics and Nuffield College, University of Oxford,
UK}\\
}
\maketitle
\begin{abstract}
\setlength{\emergencystretch}{2em}
\noindent\baselineskip=15pt For subvector inference in the linear instrumental variables model under homoskedasticity but allowing for weak instruments, \citet*{GuggenbergerKleiMavrQE2019} (GKM) propose a conditional subvector \citet{AndersonRubin1949} ($AR$) test that uses data-dependent critical values that adapt to the strength of the parameters not under test. This test has correct size and strictly higher power than the test that uses standard asymptotic chi-square critical values. The subvector $AR$ test is the minimum eigenvalue of a data-dependent matrix. The GKM critical value function conditions on the largest eigenvalue of this matrix. Instead, this paper considers a data-dependent critical value function that conditions on the second-smallest eigenvalue, as this eigenvalue is the appropriate indicator of weak identification. We find that the data-dependent critical value function of GKM also applies to this conditioning and show that this test has correct size and power strictly higher than the GKM test when the number of parameters not under test is larger than one. This is the best feasible test, where we consider feasible to be within the class of conditioning on a single eigenvalue. Our proposed procedure further applies to the subvector $AR$ test statistic that is robust to an approximate Kronecker product structure of conditional heteroskedasticity as proposed by \cite* {Guggenberger_Kleibergen_Mavroeidis_2024}, carrying over its power advantage to this setting as well. 
\end{abstract}
{\small\textbf{JEL Classification:}}{\small{} C12, C26.}{\small\par}

\noindent{\small\textbf{Keywords:}}{\small{} Instrumental Variables,
Weak Instrument Robust Inference, Subvector $AR$ Test, LIML, Overidentification
Test, Power.}{\small\par}

\thispagestyle{empty}

\baselineskip=20pt

\pagebreak{}

\setcounter{page}{1}

\section{Introduction}

\begingroup
\setlength{\emergencystretch}{2em}
For the linear instrumental variables model under homoskedasticity, \citet*{GuggenbergerKleiMavrQE2019} (GKM) introduced a subvector $AR$ testing procedure that uses data-dependent critical values that adapt to the identification strength of the parameters not under test. This test has correct size and strictly higher power than the test that uses standard asymptotic chi-square critical values. The subvector $AR$ test is the minimum eigenvalue of a data-dependent matrix, and the GKM critical-value function conditions on the largest eigenvalue of this non-central Wishart-distributed matrix. If this largest eigenvalue is small, then the nuisance parameters not under test suffer from a weak instruments problem. The $AR$ test based on the chi-square critical values is then undersized, and the conditional approach of GKM decreases the critical value. It does this while maintaining size control, and thus improves power compared to the standard approach.
\par
\endgroup

Whilst a small largest eigenvalue does indicate identification problems, a large largest eigenvalue does not necessarily imply strong identification. For example, the standard weak instruments test for the two-stage least squares estimator is the \citet{CraggDon1993} test. This is a rank test on the first-stage parameters, which is the \textit{minimum} eigenvalue of the first-stage concentration matrix \citep{StockYogo2005}. Under the maintained assumption of instrument validity, the equivalent in the LIML-based $AR$ subvector test is the second-smallest eigenvalue. A small second-smallest eigenvalue then indicates weak instrument problems for estimating the parameters not under test, affecting the size and power of the test.

We therefore investigate using data-dependent critical values that condition not on the largest but on the second-smallest eigenvalue. Our main finding is that we can use the \textit{same} critical-value function as derived by GKM. We obtain this result by first showing that, for $p>2$, if the $p-2$ largest eigenvalues of the $p\times p$ non-central Wishart-distributed matrix are all large and increasing to $\infty$, then the joint distribution of the two smallest eigenvalues is the same as that of a $2\times2$ Wishart matrix with a change in degrees of freedom, which is the difference between the number of instruments and the number of parameters not under test. The two eigenvalues of the associated non-centrality matrix match the two smallest eigenvalues of the $p\times p$ matrix. Therefore, in this large $p-2$ largest-eigenvalues setting, the results and tabulated critical values of GKM for different degrees of freedom apply directly, simply substituting the largest for the second-smallest eigenvalue as the conditioning variable. This implies that size is controlled in this case. Our second contribution is to show that when the $p-2$ largest eigenvalues decrease to less than $\infty$, the rejection probability decreases; hence the test has size less than or equal to its nominal size under the null in all circumstances. This second part follows from direct calculation of the joint distribution of the two smallest eigenvalues, based on the results of \citet{james1964distributions} and \citet{muirhead2009aspects}, from which the conditional distribution can be obtained. Simulations from the non-central Wishart distribution give the same results. As the direct calculations are very time-consuming, we further confirm the second finding with a range of simulations.

For $p>2$, as the critical value when conditioning on the second-smallest eigenvalue is smaller than that conditioning on the largest eigenvalue, this $AR$ subvector testing procedure has strictly higher power than the GKM test.

The distribution of the subset $AR$ test statistic is of course a function of all $p-1$ largest eigenvalues
of the non-central Wishart-distributed matrix, and other approaches that consider a larger subset of these eigenvalues could increase power, but likely at the expense of much higher computational cost (GKM, pp 496-497). Our method retains one of the main advantages of the GKM test: computational simplicity. We refer to this property as the feasibility of the test. In particular, we define feasibility as restricting conditioning to a single component of the spectral summary of identification strength, because this allows feasible inference via a known highly tractable approximate conditional density that relies solely on standard scalar-input functions. Within this computationally simple class, our approach represents the best test. By isolating the second-smallest eigenvalue, we condition on the component most informative about identification strength. Using power bounds for various identification settings where $m_W = 2$, we further highlight what we leave on the table by restricting ourselves to this feasible class. Although we show only limited potential for improvements in a low-dimensional setting, these gains could be more meaningful as dimensions grow large. Therefore, we emphasize that future work should explore ways to make higher-order conditioning sets feasible and further improve power.

In the next section, Section \ref{sec:finite_sample_analysis}, we first present the finite-sample model setup. We present the analysis of GKM and their $AR$ testing procedure, conditioning on the largest eigenvalue, in Section \ref{sec:CondLarge}. Subsequently, we present our main results for the $AR$ testing procedure, conditioning on the second-smallest eigenvalue, in Section \ref{sec:Cond2nd}. Section \ref{sec:Sims} presents simulation results, confirming the correct size and superior power properties of our proposed test. Section \ref{sec:empirical} introduces an empirical analysis highlighting the practical importance of the work introduced in this paper. Our testing procedure not only strengthens results reported in GKM, but also affects some of its conclusions. Furthermore, \cite*{Guggenberger_Kleibergen_Mavroeidis_2024} show for a model with heteroskedasticity of a Kronecker form that their conditional approach applies to a robust version of the $AR$ test statistic. In Appendix \ref{sec:Het} we show that these results also extend to our conditional approach. Section \ref{seq:Conclusions} concludes.

\section{Finite Sample Analysis}
\label{sec:finite_sample_analysis}

We follow the finite sample analysis of GKM and consider the linear model with endogenous variables $W$ and either endogenous or exogenous variables $X$, given by the equations
\begin{align}
y & =X\beta+W\gamma+\varepsilon,\label{eq:fullmod}\\
X & =Z\Pi_{X}+V_{X},\nonumber \\
W & =Z\Pi_{W}+V_{W},\nonumber 
\end{align}
where $y\in\mathbb{R}{}^{n}$, $X\in\mathbb{R}{}^{n\times m_{X}}$,
$W\in\mathbb{R}{}^{n\times m_{W}}$, the instruments $Z\in\mathbb{R}{}^{n\times k}$
and $k-m_{W}\geq1$. As in GKM, we assume that the instruments $Z$
are fixed, $Z^{\prime}Z$ is positive definite and that $\left(\varepsilon_{i},V_{Xi}^{\prime},V_{Wi}^{\prime}\right)^{\prime}\overset{i.i.d.}{\sim}\mathcal{N}\left(0,\Sigma\right)$,
$i=1,\ldots,n$, for some $\Sigma\in\mathbb{R}{}^{(m+1)\times(m+1)}$,
with $m=m_{X}+m_{W}$, and
\[
\Sigma=\left(\begin{array}{ccc}
\sigma_{\varepsilon\varepsilon} & \Sigma_{\varepsilon V_{X}} & \Sigma_{\varepsilon V_{W}}\\
\Sigma_{\varepsilon V_{X}}^{\prime} & \Sigma_{V_{X}V_{X}} & \Sigma_{V_{X}V_{W}}\\
\Sigma_{\varepsilon V_{W}}^{\prime} & \Sigma_{V_{X}V_{W}}^{\prime} & \Sigma_{V_{W}V_{W}}
\end{array}\right).
\]

For testing the hypothesis
\[
H_{0}:\beta=\beta_{0}\,\,\,\text{against}\,\,\,H_{1}:\beta\neq\beta_{0},
\]
let
\[
y_{0}:=y-X\beta_{0},
\]
and consider the restricted model given by the equations 
\begin{align}
y_{0} & =W\gamma+\varepsilon\left(\beta_{0}\right),\label{eq:restrmod}\\
W & =Z\Pi_{W}+V_{W},\nonumber 
\end{align}
with $\varepsilon\left(\beta_{0}\right)=\varepsilon+X\left(\beta-\beta_{0}\right)$.

The reduced form for $y_{0}$ is given by
\[
y_{0}=Z\pi_{y_{0}}+v_{y_{0}},
\]
with $\pi_{y_{0}}=\Pi_{X}\left(\beta-\beta_{0}\right)+\Pi_{W}\gamma$
and $v_{y_{0}}=\varepsilon+V_{X}\left(\beta-\beta_{0}\right)+V_{W}\gamma$.
It follows that $\left(v_{y_{0}i},V_{Wi}^{\prime}\right)\overset{i.i.d.}{\sim}\,\mathcal{N}\left(0,\Omega\left(\beta_{0}\right)\right)$,
with
\begin{align*}
\Omega\left(\beta_{0}\right) & =\left(\begin{array}{cc}
1 & 0\\
\beta-\beta_{0} & 0\\
\gamma & I_{m_{W}}
\end{array}\right)^{\prime}\Sigma\left(\begin{array}{cc}
1 & 0\\
\beta-\beta_{0} & 0\\
\gamma & I_{m_{W}}
\end{array}\right),
\end{align*}
where we assume, as in GKM, that $\Omega\left(\beta_{0}\right)$ is
known and positive definite.

The subvector \citet{AndersonRubin1949} test statistic is then defined
as 
\begin{align*}
AR_{n}\left(\beta_{0}\right) & \coloneqq \min_{\tilde{\gamma}\in\mathbb{R}^{m_{W}}}\frac{\left(y_{0}-W\tilde{\gamma}\right)^{\prime}P_{Z}\left(y_{0}-W\tilde{\gamma}\right)}{\left(1,-\tilde{\gamma}^{\prime}\right)\Omega\left(\beta_{0}\right)\left(1,-\tilde{\gamma}^{\prime}\right)^{\prime}}\\
 & =\min\,\text{eval}\left(\left(y_{0},W\right)^{\prime}P_{Z}\left(y_{0},W\right)\left(\Omega\left(\beta_{0}\right)\right)^{-1}\right),
\end{align*}
where $P_Z=Z(Z'Z)^{-1}Z'$ and $\min\,\text{eval}\left(A\right)$ denotes the minimum eigenvalue
of $A$. The resulting estimator for $\tilde{\gamma}$ is the so-called
LIMLK estimator of $\gamma$ in the restricted model (\ref{eq:restrmod}),
an infeasible estimator which requires the reduced form error covariance
matrix to be known, see \citet{Anderson1977} and the discussion in
\citet[p 564]{StaigerStockEcta1997}.\footnote{The feasible statistic as evaluated later in Section \ref{sec:Sims} is given
by $AR_{n,f}\left(\beta_{0}\right)=\frac{\left(y_{0}-W\hat{\tilde{\gamma}}_{L}\right)^{\prime}P_{Z}\left(y_{0}-W\hat{\tilde{\gamma}}_{L}\right)}{\left(y_{0}-W\hat{\tilde{\gamma}}_{L}\right)^{\prime}M_{Z}\left(y_{0}-W\hat{\tilde{\gamma}}_{L}\right)/\left(n-k\right)}$,
where $\hat{\tilde{\gamma}}_{L}$ is the LIML estimator of $\gamma$
in the restricted model (\ref{eq:restrmod}) and $M_Z=I_n-P_Z$}

Denote the $p:=m_{W}+1$ eigenvalues of the matrix $\left(y_{0},W\right)^{\prime}P_{Z}\left(y_{0},W\right)\left(\Omega\left(\beta_{0}\right)\right)^{-1}$
by $\hat{\kappa}_{1}\geq\hat{\kappa}_{2}\geq\ldots\geq\hat{\kappa}_{p}$,
so that $AR_{n}\left(\beta_{0}\right)=\hat{\kappa}_{p}$. GKM show
that the eigenvalues $\hat{\kappa}_{j}$ for $j=1,\ldots,p$ are equal
to the eigenvalues of a $p\times p$ non-central Wishart matrix $\Xi^{'}\Xi$,
with the $k$ rows of the $k\times p$ matrix $\Xi$ are independently
normally distributed with common covariance matrix $I_{p}$ and $\mathbb{E}\left(\Xi\right)=\mathcal{M}$.
Therefore, $\Xi^{'}\Xi\sim\mathcal{W}_{p}\left(k,I_{p},\mathcal{M}^{\prime}\mathcal{M}\right)$,
a non-central Wishart distribution with $k$ degrees of freedom, covariance
matrix $I_{p}$ and noncentrality matrix $\mathcal{M}^{\prime}\mathcal{M}$.

As GKM show, under the null $H_{0}:\beta=\beta_{0}$, $\mathcal{M}=\left(0,\Theta_{W}\right)$,
where the $k\times m_{W}$ matrix $\Theta_{W}$ is given by
\[
\Theta_{W}:=\left(Z^{\prime}Z\right)^{1/2}\Pi_{W}\Sigma_{V_{W}V_{W}.\varepsilon}^{-1/2},
\]
where 
\[
\Sigma_{V_{W}V_{W}.\varepsilon}=\Sigma_{V_{W}V_{W}}-\Sigma_{\varepsilon V_{W}}^{\prime}\Sigma_{\varepsilon V_{W}}\sigma_{\varepsilon\varepsilon}^{-1}.
\]
It follows that
\[
\mathcal{M}^{\prime}\mathcal{M}=\left(\begin{array}{cc}
0 & 0\\
0 & \Theta_{W}^{\prime}\Theta_{W}
\end{array}\right),
\]
with
\[
\Theta_{W}^{\prime}\Theta_{W}=\Sigma_{V_{W}V_{W}.\varepsilon}^{-1/2}\Pi_{W}^{\prime}Z^{\prime}Z\Pi_{W}\Sigma_{V_{W}V_{W}.\varepsilon}^{-1/2}.
\]
The joint distribution of $\left(\hat{\kappa}_{1},\ldots,\hat{\kappa}_{p}\right)$
under the null then only depends on the eigenvalues of $\Theta_{W}^{\prime}\Theta_{W}$,
which GKM denote by
\[
\kappa_{j}:=\kappa_{j}\left(\Theta_{W}^{\prime}\Theta_{W}\right),\,\,\,\,j=1,\ldots,m_{W}.
\]

\subsection{Concentration Parameter}
It is at this point interesting to compare the $\Theta_{W}^{\prime}\Theta_{W}$
concentration matrix to the standard concentration matrix, which is
defined using the first-stage model for $W$ and given
by $\Sigma_{V_{W}V_{W}}^{-1/2}\Pi_{W}^{\prime}Z^{\prime}Z\Pi_{W}\Sigma_{V_{W}V_{W}}^{-1/2}$ (\citet*{StockWrightYogoJBES2002}, \citet{StockYogo2005}).
For the $m_{W}=1$ case, and writing the variance matrix under the
null in the restricted model (\ref{eq:restrmod}) as $\Sigma_{r}=\left(\begin{array}{cc}
\sigma_{\varepsilon}^{2} & \sigma_{\varepsilon v_{w}}\\
\sigma_{\varepsilon v_{w}} & \sigma_{v_{w}}^{2}
\end{array}\right)$ and the correlation coefficient $\rho_{\varepsilon v_{w}}=\frac{\sigma_{\varepsilon v_{w}}}{\sigma_{\varepsilon}\sigma_{v_{w}}}$,
the standard concentration parameter is given by
\[
\mu_{w}^{2}=\frac{\pi_{w}^{\prime}Z^{\prime}Z\pi_{w}}{\sigma_{v_{w}}^{2}},
\]
whereas
\[
\theta_{w}^{\prime}\theta_{w}=\frac{\pi_{w}^{\prime}Z^{\prime}Z\pi_{w}}{\sigma_{v_{w}}^{2}-\frac{\sigma_{\epsilon v_{w}}^{2}}{\sigma_{\varepsilon}^{2}}}=\frac{\pi_{w}^{\prime}Z^{\prime}Z\pi_{w}}{\sigma_{v_{w}}^{2}\left(1-\rho_{\varepsilon v_{w}}^{2}\right)}.
\]
$\mu_{w}^{2}$ is the noncentrality parameter for the Wald test, testing
$H_{0}:\pi_{w}=0$ based on the OLS estimator $\hat{\pi}_{w}=\left(Z^{\prime}Z\right)^{-1}Z^{\prime}w$.
Moreover, $\theta_{w}^{\prime}\theta_{w}$ is the noncentrality parameter
for the Wald test based on, here, the infeasible LIMLK estimator of
$\pi_{w}$ in the restricted model (\ref{eq:restrmod}). As this is
essentially part of a FIML-type system estimator, the structural errors
enter this expression.

Under the weak instrument, local-to-zero setting $\pi_{w}=\pi_{w,n}=c/\sqrt{n}$,
with $c$ a vector of constants, it follows that $\theta_{w}^{\prime}\theta_{w}=$$\frac{c^{\prime}\left(Z^{\prime}Z/n\right)c}{\sigma_{v_{w}}^{2}\left(1-\rho_{\varepsilon v_{w}}^{2}\right)}$.
So, assuming that $c^{\prime}\left(Z^{\prime}Z/n\right)c<a$,
$\forall n$, for a finite $a\in\mathbb{R}^{+}$, under weak instruments,
this noncentrality/concentration parameter $\theta_{w}^{\prime}\theta_{w}=\kappa_{1}\rightarrow\infty$
when $\rho_{\varepsilon v_{w}}^{2}\rightarrow1$. Hence, in that
case the weak instrument distribution of the $AR_{n}\left(\beta_{0}\right)$
statistic under the null, is the standard strong instrument $\chi_{k-1}^{2}$
distribution, following the results of Theorem 1 in GKM; see also
\citet{VandeSijpe2023}.

\subsection{Conditioning on Largest Eigenvalue}
\label{sec:CondLarge}

\subsubsection{\texorpdfstring{$m_{W}=1$}{m\_W = 1}}
\label{sec:Cond21}

For the $m_{W}=1$ case, GKM show that an approximate conditional
density of the smallest eigenvalue, $\hat{\kappa}_{2}$, given the
largest, $\hat{\kappa}_{1}$, is given by
\begin{equation}
    \label{eq:cpdf}
    f_{\hat{\kappa}_{2}|\hat{\kappa}_{1}}^{*}\left(x_{2}|\hat{\kappa}_{1}\right)=f_{\chi_{k-1}^{2}}\left(x_{2}\right)\left(\hat{\kappa}_{1}-x_{2}\right)^{1/2}g\left(\hat{\kappa}_{1}\right),\,\,\,x_{2}\in\left[0,\hat{\kappa}_{1}\right],
\end{equation}
where $f_{\chi_{k-1}^{2}}\left(.\right)$ is the density of a $\chi_{k-1}^{2}$ distribution
and $g\left(\hat{\kappa}_{1}\right)$ is a function that does not
depend on any unknown parameters. Using this conditional density function,
GKM derive and tabulate data dependent, conditional on $\hat{\kappa}_{1}$,
critical values, $c_{1-\alpha}\left(\hat{\kappa}_{1},k-1\right)$, for
the subvector $AR_{n}\left(\beta_{0}\right)$ statistic. As for this setting we have $AR_{n}\left(\beta_{0}\right)=\hat{\kappa}_{2}$,
the test function for rejecting $H_{0}$ at nominal size $\alpha$
is written as
\begin{equation}
\label{eq:phic}
\phi_{c}\left(\hat{\kappa}_{1,2}\right):=\mathbf{1}\left[\hat{\kappa}_{2}>c_{1-\alpha}\left(\hat{\kappa}_{1},k-1\right)\right],
\end{equation}
where \textbf{$\mathbf{1}\left[.\right]$} is the indicator function, and for general $s,r$, we write $\hat{\kappa}_{s,r}:=\left\{ \hat{\kappa}_{s},\hat{\kappa}_{r}\right\} $.

GKM (Theorem 2, p 496) show by numerical integration and Monte Carlo simulation that the $AR$ test procedure defined in (\ref{eq:phic}) has correct size for $\alpha\in\{0.10,0.05,0.01\}$ and $k-m_W\in\{1,\ldots,20\}$.

As $c_{1-\alpha}\left(\hat{\kappa}_{1},k-1\right)<c_{1-\alpha}\left(\infty,k-1\right)$, $\phi_{c}$ has strictly higher power than the unconditional testing procedure, $\phi_{\chi^{2}}$ , that uses the $1-\alpha$ critical values of the $\chi_{k-1}^{2}$
distribution as discussed in \citet*{GugKleiMavrChenEcta2012}.

\subsubsection{\texorpdfstring{$m_{W}>1$}{m\_W > 1}}

The approach by GKM for $m_{W}>1$ is to use the following test function
\[
\phi_{c_{1}}\left(\hat{\kappa}_{1,p}\right):=\mathbf{1}\left[\hat{\kappa}_{p}>c_{1-\alpha}\left(\hat{\kappa}_{1},k-m_{W}\right)\right],
\]
where the subscript in $c_{1}$ makes it clear that the conditioning is on
the largest eigenvalue. GKM (Corollary 4, p 500) show that $\phi_{c_{1}}$ has correct
size for the tabulated critical value functions as mentioned in Section \ref{sec:Cond21}, with $\alpha\in\{0.10,0.05,0.01\}$ and $k-m_W\in\{1,\ldots,20\}$, and has power strictly higher than $\phi_{\chi^{2}}$, which uses the critical values of the $\chi_{k-m_{W}}^{2}$ distribution.

\subsection{Conditioning on Second-Smallest Eigenvalue}
\label{sec:Cond2nd}

Instead of conditioning on the largest eigenvalue, we propose
conditioning on the second-smallest eigenvalue, $\hat{\kappa}_{m_{W}}$,
with the test function given by
\[
\phi_{c_{p-1}}\left(\hat{\kappa}_{p-1,p}\right):=\mathbf{1}\left[\hat{\kappa}_{p}>c_{1-\alpha}\left(\hat{\kappa}_{p-1},k-m_{W}\right)\right].
\]
The reason for considering conditioning on the second-smallest eigenvalue
is that the smallest eigenvalue of the concentration matrix $\Theta_{W}^{\prime}\Theta_{W}$
is a better indicator of the identification strength of $\gamma$
than the largest eigenvalue. This is also why the standard \citet{CraggDon1993}
rank test for underidentification is based on the minimum eigenvalue of
the concentration matrix $\hat{\Sigma}_{V_{W}V_{W}}^{-1/2}\hat{\Pi}_{W}^{\prime}Z^{\prime}Z\hat{\Pi}_{W}\hat{\Sigma}_{V_{W}V_{W}}^{-1/2}$,
with $\hat{\Pi}_{W}=\left(Z^{\prime}Z\right)^{-1}Z^{\prime}W$ and
$\hat{\Sigma}_{V_{W}V_{W}}=W^{\prime}M_{Z}W/\left(n-k\right)$ \citep{StockYogo2005}.

For example, for $m_{W}=2$, we can have the situation that $\kappa_{1}=\infty$,
but $\kappa_{2}$ is small, in which case $\gamma$ is weakly identified.
Conditioning on $\hat{\kappa}_{1}$ would then result in conditional
critical values close to those of $\chi_{k-m_{W}}^{2}$, whereas
conditioning on $\hat{\kappa}_{2}$ would properly address the weak
identification problem.

\sloppy As $\kappa_{1}\geq\kappa_{2}\geq...\geq\kappa_{p}$, we have
that if $\kappa_{p-2}=\infty$, it follows that also $\kappa_{j}=\infty$,
$j=1,\ldots,p-3$. For $\Xi^{\prime}\Xi\sim\mathcal{W}_{p}\left(k,I_{p},\mathcal{M}^{\prime}\mathcal{M}\right)$,
our main finding is that the joint distribution of $\left(\hat{\kappa}_{p-1},\hat{\kappa}_{p}\right)$,
when the eigenvalues of $\mathcal{M}^{\prime}\mathcal{M}$ are given
by $\left\{ \kappa_{1}=\ldots=\kappa_{p-2}=\infty,\kappa_{p-1},\kappa_{p}=0\right\} $,
is the same as that of $\left(\hat{\kappa}_{1}^{*},\hat{\kappa}_{2}^{*}\right)$,
which are the eigenvalues of $\Xi^{*\prime}\Xi^{*}\sim\mathcal{W}_{2}\left(k-m_{W}+1,I_{2},\mathcal{M^{*}}^{\prime}\mathcal{M}^{*}\right)$,
with the two eigenvalues of $\mathcal{M^{*}}^{\prime}\mathcal{M}^{*}$
equal to $\kappa_{1}^{*}=\kappa_{p-1}$ and $\kappa_{2}^{*}=0$. It
therefore follows that
\[
f_{\hat{\kappa}_{p}|\hat{k}_{p-1};\kappa_{p-2}=\infty}^{*}\left(x_{p}|\hat{\kappa}_{p-1};\kappa_{p-2}=\infty\right)=f_{\chi_{k-m_{W}}^{2}}\left(x_{p}\right)\left(\hat{\kappa}_{p}-x_{p}\right)^{1/2}g\left(\hat{\kappa}_{p-1}\right),\,\,\,x_{p}\in\left[0,\hat{\kappa}_{p-1}\right].
\]
Therefore, we have that the conditional critical values of GKM apply directly to this
conditioning, where we only have to adjust the degrees of freedom commensurate with
$m_{W}$. We state this result formally in the following Proposition.

\begin{prop}
\label{Prop infty} Let $\Xi^{\prime}\Xi\sim\mathcal{W}_{p}\left(k,I_{p},\mathcal{M}^{\prime}\mathcal{M}\right)$,
with $p=m_{W}+1>2$, and let $\hat{\kappa}_{1}\geq\hat{\kappa}_{2}\geq\ldots\geq\hat{\kappa}_{p}$
denote the ordered eigenvalues of $\Xi^{\prime}\Xi$. The joint distribution
of $\left(\hat{\kappa}_{p-1},\hat{\kappa}_{p}\right)$ when the eigenvalues
of $\mathcal{M}^{\prime}\mathcal{M}$ are given by $\left\{ \kappa_{1}=\ldots=\kappa_{p-2}=\infty,\kappa_{p-1},\kappa_{p}=0\right\} $
is the same as that of $\left(\hat{\kappa}_{1}^{*},\hat{\kappa}_{2}^{*}\right)$,
which are the eigenvalues of $\Xi^{*\prime}\Xi^{*}\sim\mathcal{W}_{2}\left(k-m_{W}+1,I_{2},\mathcal{M^{*}}^{\prime}\mathcal{M}^{*}\right)$,
with the two eigenvalues of $\mathcal{M^{*}}^{\prime}\mathcal{M}^{*}$
equal to $\kappa_{1}^{*}=\kappa_{p-1}$ and $\kappa_{2}^{*}=0$. 
\end{prop}
\begin{proof}
An analytical proof is given in Appendix \ref{subsec:Proof-of-Proposition}. 
\end{proof}
It follows directly from Proposition \ref{Prop infty} that by simply replacing the values of $\hat{\kappa}_{1}$ with $\hat{\kappa}_{p-1}$, the critical value functions tabulated by GKM for different values of $\alpha$ and $k-m_{W}$ apply directly to this setting with $\kappa_{1}=\ldots=\kappa_{p-2}=\infty$.

We illustrate the result of Proposition \ref{Prop infty} by a simulation exercise. We draw $1,000,000$ samples from the Wishart distribution, $\Xi^{\prime}\Xi\sim\mathcal{W}_{p}\left(k,I_{p},\mathcal{M}^{\prime}\mathcal{M}\right)$,
with 
\[
\mathcal{M}^{\prime}\mathcal{M}=\left(\begin{array}{ccccc}
0 & 0 &  &  & 0\\
0 & \kappa_{1} & 0\\
 & 0 & \kappa_{2} & 0\\
 &  & 0 & \ddots & 0\\
0 &  &  & 0 & \kappa_{p-1}
\end{array}\right),
\] and present the null rejection
probabilities (NRPs) for the $AR$ test procedures $\phi_{\chi^{2}}$,
$\phi_{c_{1}}$ and $\phi_{c_{p-1}}$ at 5\% level, as a function of
$\kappa_{p-1}$. For each $m_{W}>1$, we set $\kappa_{1}=\ldots=\kappa_{p-2}=100,000$. These results are presented in Figure \ref{fig:NRPlarge}. It
is clear that the behaviour of $\phi_{c_{p-1}}$ is the same for all
values of $m_{W}$, whereas here the behaviours of $\phi_{\chi^{2}}$
and $\phi_{c_{1}}$ are the same for $m_{W}>1$, as $\kappa_{1}$
is set very large.\footnote{\label{fnotekap} Here and below, the NRPs are evaluated at 43 values for $\kappa_{p-1}$ given by $\exp(0,s,2s,\ldots,\log(101))-1$, where $s=\log(101)/42$.}

\begin{figure}[H]
\begin{centering}
\includegraphics[scale=0.60]{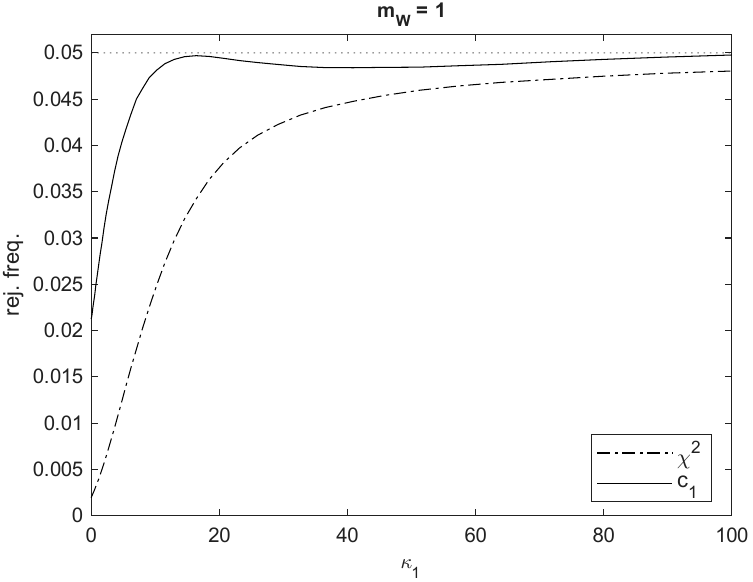}\includegraphics[scale=0.60]{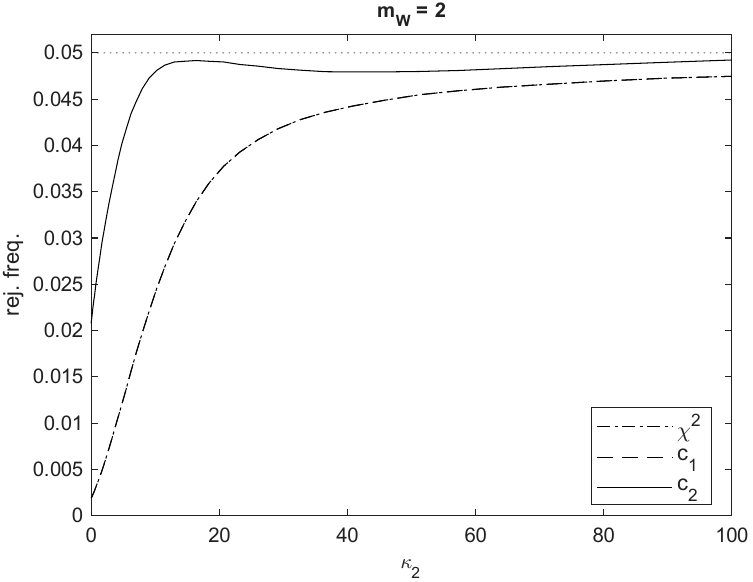}
\par\end{centering}
\begin{centering}
\includegraphics[scale=0.60]{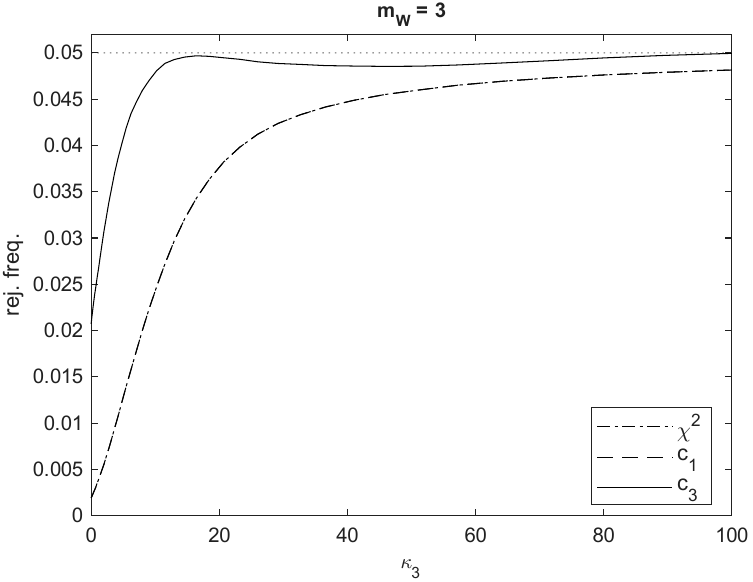}
\par\end{centering}
\caption{\label{fig:NRPlarge} Null rejection probabilities at 5\% level of
$\phi_{\chi^{2}}$, $\phi_{c_{1}}$ and $\phi_{c_{p-1}}$ for different
values of $m_{W}$ as a function of $\kappa_{p-1}$, with $k-m_{W}=4$.
For $m_{W}>1$, $\kappa_{j}=100,000$ for $j<m_{W}$. Critical values
for $\phi_{c_{1}}$ and $\phi_{c_{p-1}}$ obtained from Table 6 in
GKM, together with linear interpolation, with for $\phi_{c_{p-1}}$ the conditioning largest eigenvalue $\hat{\kappa}_{1}$ substituted by the second-smallest one, $\hat{\kappa}_{p-1}$.}
\end{figure}

\subsection{Size Control}

From Proposition \ref{Prop infty} it follows from Theorem 2 in GKM (p 496) that $\phi_{c_{p-1}}$ has correct size for all values of $m_W$, for the values of $\alpha \in \{0.10,0.05,0.01\}$, $k-m_W \in \{1,\ldots,20\}$ and when $\kappa_1=\kappa_2=\ldots=\kappa_{p-2}=\infty$, as illustrated in Figure \ref{fig:NRPlarge}. We next find by direct calculations and simulations that the test procedure $\phi_{c_{p-1}}$ controls size also for the cases where all or some of the $\kappa_{j}$, $j=1,\ldots,p-2$, are less then $\infty$. It follows therefore that $NRP_{\phi_{c_{p-1}}}\leq\alpha$ for the general case of $\alpha \in \{0.10,0.05,0.01\}$ and $k-m_W \in \{1,\ldots,20\}$.

In Appendix \ref{subsec:Eigenvalues}, we show how to calculate
the joint probability density function (pdf) of the ordered eigenvalues
$\hat{\kappa}_{1}\geq\hat{\kappa}_{2}\geq\ldots\geq\hat{\kappa}_{p}$,
from which we can calculate the conditional pdf $f_{\hat{\kappa}_{p}|\hat{\kappa}_{p-1}}\left(x_{p}|\hat{\kappa}_{p-1}\right)$. Table \ref{tab:CondCDF}
presents some calculated values of the conditional cdf for $p=3$ and $k=10$ (hence $m_W=2$ and $k-m_W=8$),
\[
\mathbb{P}\left(\hat{\kappa}_{3}\leq a|\hat{\kappa}_{2}=10;\kappa_{1},\kappa_{2}=2,\kappa_{3}=0\right),
\]
for the values $a=\left(4,5,6\right)$, and different values of $\kappa_{1}=\left(2,5,10,\infty\right)$.
It is clear that the conditional cdfs are monotonically decreasing
for increasing values of $\kappa_{1}$. 

\begin{table}[H]
\begin{centering}
\begin{tabular}{c|ccc}
\hline 
 & \multicolumn{3}{c}{$a$}\tabularnewline
\hline 
$\kappa_{1}$ & $4$ & $5$ & $6$\tabularnewline
\hline 
$2$ & $0.37$ & $0.56$ & $0.73$\tabularnewline
$5$ & $0.35$ & $0.54$ & $0.72$\tabularnewline
$10$ & $0.34$ & $0.53$ & $0.70$\tabularnewline
$\infty$ & $0.20$ & $0.36$ & $0.54$\tabularnewline
\hline 
\end{tabular}
\par\end{centering}
\caption{\label{tab:CondCDF} Conditional cdf, $\mathbb{P}\left(\hat{\kappa}_{3}\protect\leq a|\hat{\kappa}_{2}=10;\kappa_{1},\kappa_{2}=2,\kappa_{3}=0\right)$,
$p=3$, $k=10$, for different values of $\kappa_{1}$ and $a$.}
\end{table}

\sloppy These numerical calculations of the conditional cdfs are replicated
by simulations. We draw larges samples from the Wishart distribution as described in Section \ref{sec:Cond2nd}
and obtain from these samples the conditional cdfs as the observed
frequencies $P\left(\hat{\kappa}_{p}\leq a|\hat{\kappa}_{p-1}=b\pm h\right)$,
for a small bandwidth $h=0.1$.\footnote{The code used for these simulations, along with the power bound calculations and empirical analysis, is available at \url{https://github.com/jesse-hoekstra/subvector\_AR\_HW}.} Figure \ref{fig:Ccdf} shows the results
for the same setup as in Table 1. The monotonicity is again confirmed
across the range $a\in[0,10]$. The right panel further shows that the points following from exact calculations of the conditional cdf as indicated in Table \ref{tab:CondCDF}
align exactly with the simulated ones. 

For this $p=3$, $k=10$ example we therefore find that
\[
\mathbb{P}\left(\hat{\kappa_3}\leq a|\hat{\kappa}_{2}=10\right)_{\kappa_{1}<\infty,\kappa_2=2}>\mathbb{P}\left(\hat{\kappa}_{3}\leq a|\hat{\kappa}_{2}=10\right)_{\kappa_{1}=\infty,\kappa_2=2}
\]
\[
\mathbb{P}\left(\hat{\kappa_3}>a|\hat{\kappa}_{2}=10\right)_{\kappa_{1}<\infty,\kappa_2=2}\leq \mathbb{P}\left(\hat{\kappa}_{3}>a|\hat{\kappa}_{2}=10\right)_{\kappa_{1}=\infty,\kappa_2=2}.
\]
Then, combining this with the findings from GKM
it shows that, under the null $H_{0}:\beta=\beta_{0}$, 
\[
\mathbb{P}\left(AR_{n}\left(\beta_{0}\right)>c_{1-\alpha}\left(\hat{\kappa}_{2}=10,k-m_{W}=8\right)\right)_{\kappa_{1}<\infty,\kappa_2=2} \\
\]
\[
\leq\mathbb{P}\left(AR_{n}\left(\beta_{0}\right)>c_{1-\alpha}\left(\hat{\kappa}_{2}=10,k-m_{W}=8\right)\right)_{\kappa_{1}=\infty,\kappa_2=2}\leq\alpha.
\]

\begin{figure}[H]
\begin{centering}
\noindent
\includegraphics[width=0.49\linewidth]{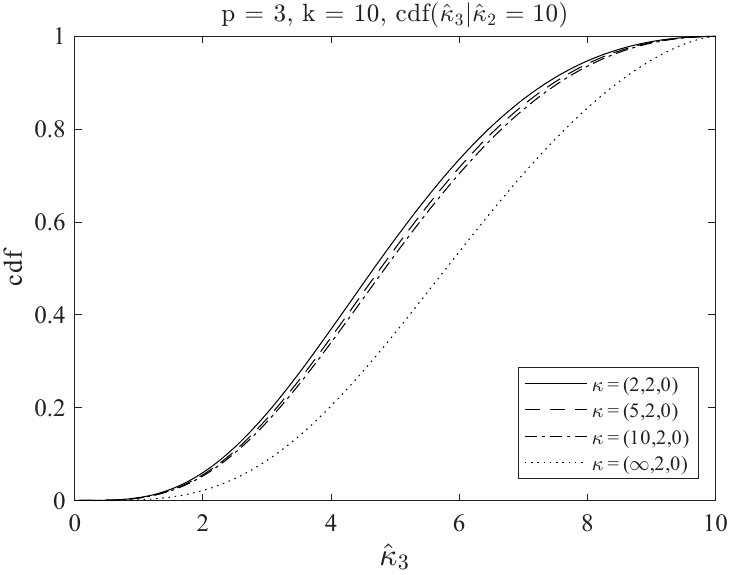}\hfill%
\includegraphics[width=0.49\linewidth]{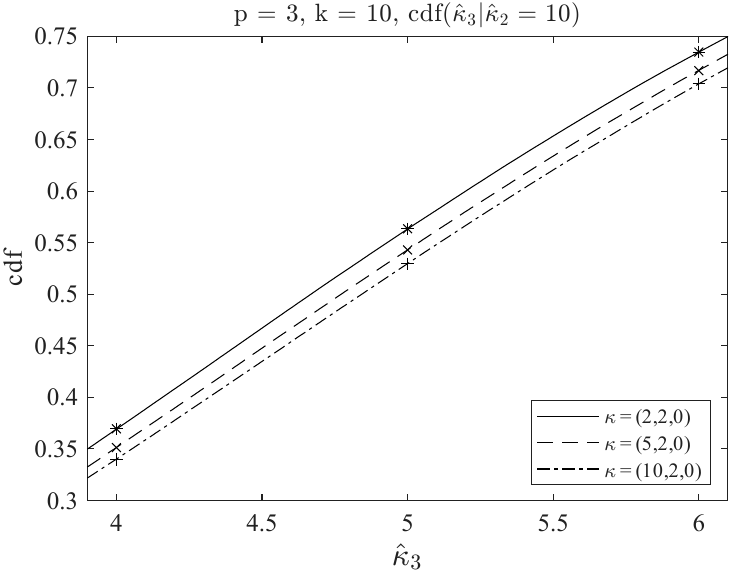}
\par\end{centering}
\caption{\label{fig:Ccdf} Left panel, simulated conditional cdf. Right panel
including the calculated points as in Table \ref{tab:CondCDF}.}
\end{figure}

In Appendix \ref{subsec:Eigenvalues},
we show in Figures \ref{fig:Ccdfapp3} - \ref{fig:Ccdfapp5&8} simulated conditional cdfs for a range of settings, for $p=3$, $p=4$, $p=5$ and $p=8$. All these examples show the same behaviour with the conditional cdfs monotonically decreasing for increasing values of $\kappa_1,\ldots,\kappa_{p-2}$.

We further show the size control of $\phi_{c_{p-1}}$ by an extensive set of simulations of the null rejection probabilities of the test. 
In Figure \ref{fig:NRPmw2} we show the NRPs of the tests $\phi_{\chi^2}$, $\phi_{c_{1}}$ and $\phi_{c_{p-1}}$ for
$m_{W}=2$ and $k-m_W=4$ at the 5\% level. In the left panel as a function of $\kappa_{1}$
for a fixed value $\kappa_{2}=10$ and in the right panel as a function of $\kappa_{2}$ with $\kappa_{1}=\kappa_{2}$. We see that the rejection frequencies
of the tests are monotonically increasing in $\kappa_{1}$ when keeping
$\kappa_{2}$ fixed and the NRPs in both panels are hence below those when $\kappa_{1}=100,000$ as displayed in Figure \ref{fig:NRPlarge}.

This is further shown in Figure \ref{fig:NRPmw2k1k2a5}, which displays the simulated NRPs for $\phi_{c_{p-1}}$ at the 5\% level for $m_W=2$ and $k-m_W=4$ for the range of values of both $\kappa_1$ and $\kappa_2$ as in footnote \ref{fnotekap}. The figure shows that for each value of $\kappa_2$ the NRPs of the test are non-decreasing for increasing values of $\kappa_1$. As in the limit for $\kappa_1\rightarrow\infty$ the test has correct size, the NRPs of the test are less than $\alpha=0.05$ everywhere in the graph.

We have repeated the simulations of the NRPs as in Figure \ref{fig:NRPmw2k1k2a5} for the cases with $m_W=2$, $k-m_W=1,\ldots,20$ and for $\alpha=0.10,0.05,0.01$. All graphs are presented in the Supplementary Material. They all display the same behaviour of the NRPs as those given in Figure \ref{fig:NRPmw2k1k2a5}, and in all cases the NRPs are less than the sizes $\alpha$ considered.

\begin{figure}[H]
\begin{centering}
\noindent
\includegraphics[width=0.49\linewidth]{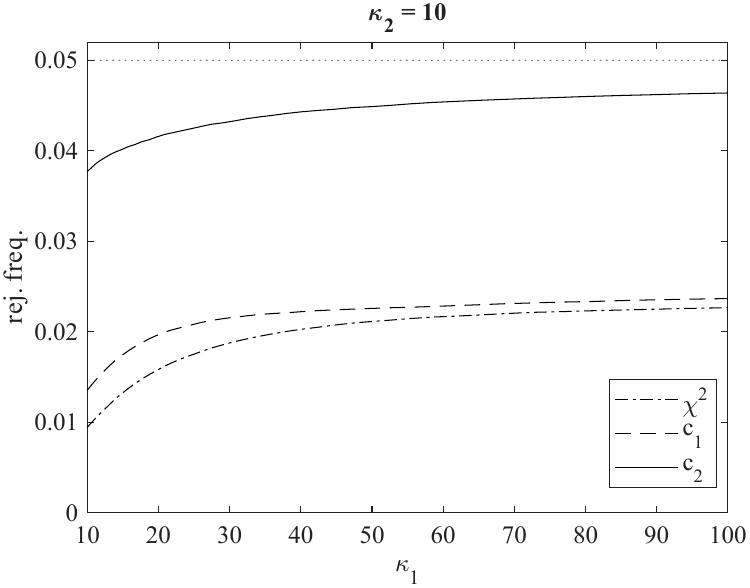}\hfill%
\includegraphics[width=0.49\linewidth]{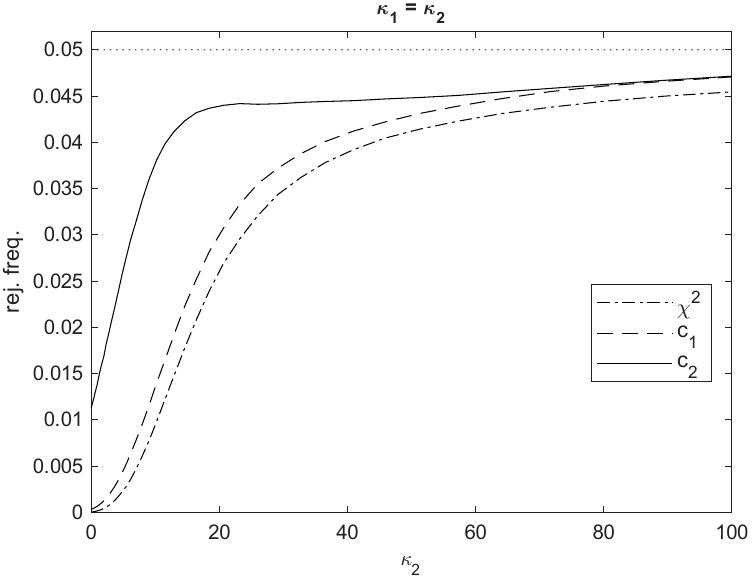}
\par\end{centering}
\caption{\label{fig:NRPmw2} NRPs at 5\% level of $\phi_{\chi^{2}}$, $\phi_{c_{1}}$
and $\phi_{c_{p-1}}$ for $m_{W}=2$ and $k-m_{W}=4$. Left
panel as a function of $\kappa_{1}$ for a fixed value of $\kappa_{2}=10$. Right panel as a function of $\kappa_{2}$, with $\kappa_{1}=\kappa_{2}$. 
See further notes to Figure \ref{fig:NRPlarge} and footnote \ref{fnotekap}.}
\end{figure}

\begin{figure}[H]
\begin{centering}
\includegraphics[scale=0.40]{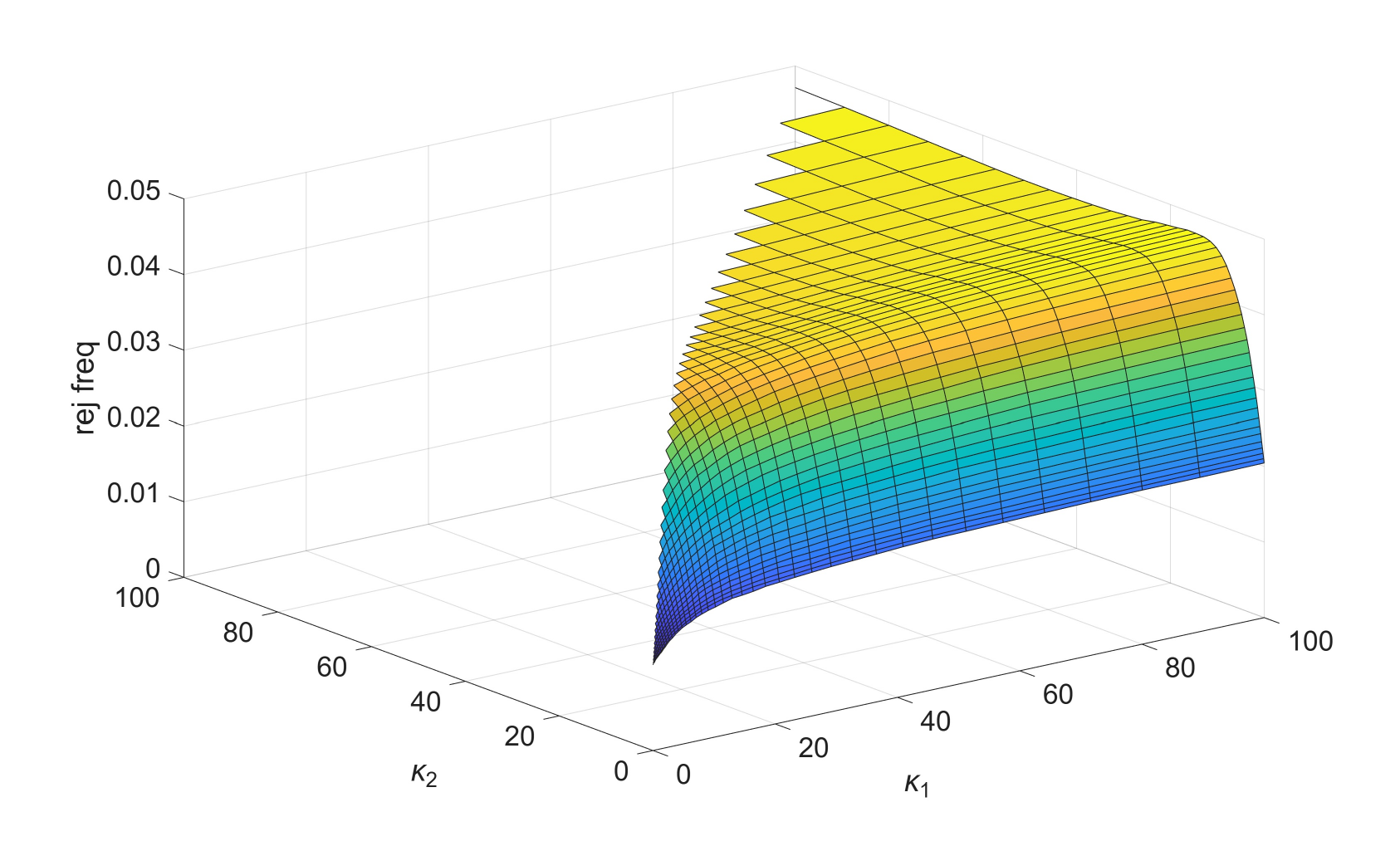}
\par\end{centering}
\caption{\label{fig:NRPmw2k1k2a5} NRPs at 5\% level of $\phi_{c_{2}}$ for $m_{W}=2$ and $k-m_{W}=4$, for various values of $\kappa_1$ and $\kappa_2$. See further notes to Figure \ref{fig:NRPlarge} and footnote \ref{fnotekap}.}
\end{figure}

\begin{figure}[H]
\begin{centering}
\includegraphics[scale=0.60]{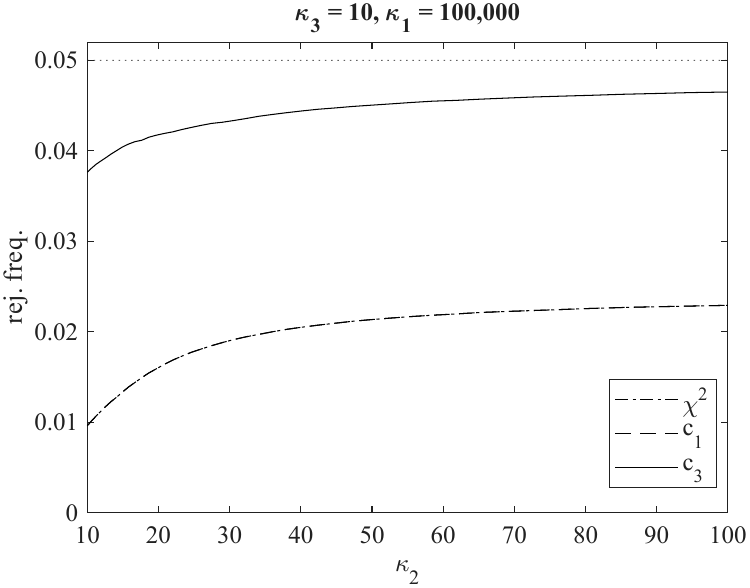}\includegraphics[scale=0.60]{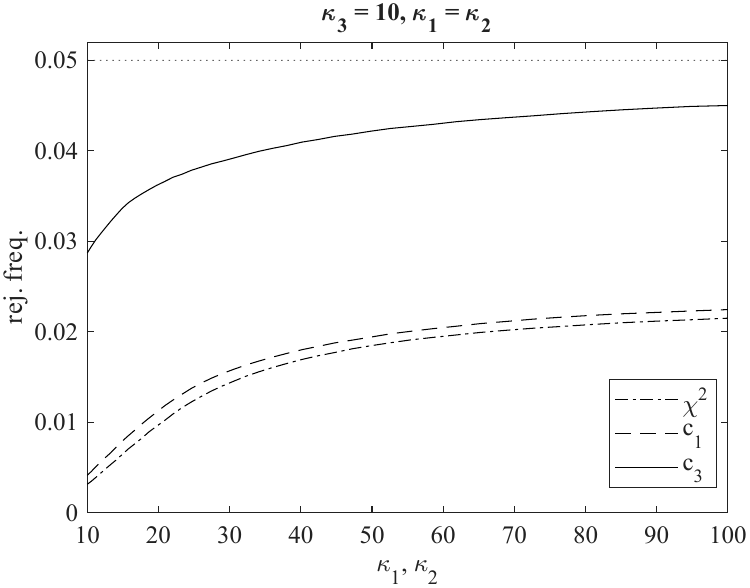}
\par\end{centering}
\begin{centering}
\includegraphics[scale=0.60]{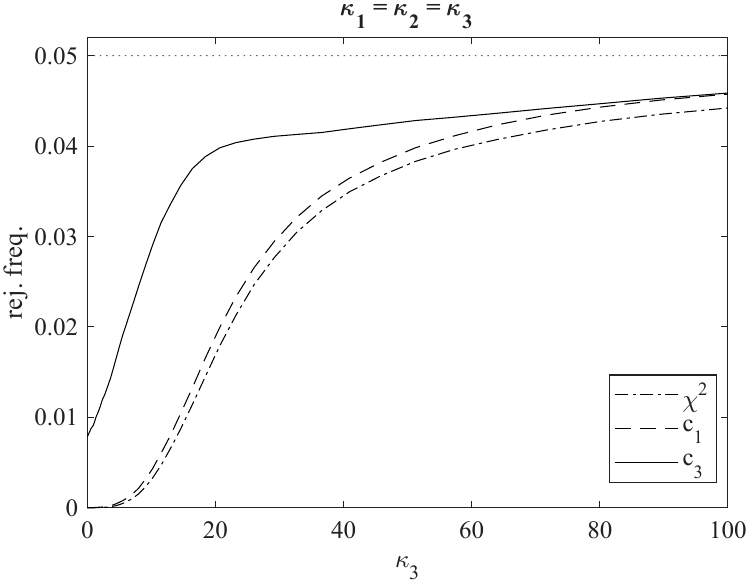}
\par\end{centering}
\caption{\label{fig:NRPmw3} NRPs at 5\% level of $\phi_{\chi^{2}}$, $\phi_{c_{1}}$
and $\phi_{c_{p-1}}$ for $m_{W}=3$ and $k-m_{W}=4$. First graph as a function of $\kappa_{2}$ for a fixed value of $\kappa_{3}=10$, with $\kappa_{1}=100,000$.
Second graph as a function of $\kappa_{1}=\kappa_{2}$,
for a fixed value of $\kappa_{3}=10$. 
 Third graph  as
a function of $\kappa_{3}$, with $\kappa_{1}=\kappa_{2}=\kappa_{3}$. See further notes to Figure \ref{fig:NRPlarge} and footnote \ref{fnotekap}.}
\end{figure}

Next we show some simulated NRPs for the tests $\phi_{\chi^2}$, $\phi_{c_{1}}$ and $\phi_{c_{p-1}}$ for $m_{W}=3$ and $k-m_W=4$ at the 5\% level in Figure \ref{fig:NRPmw3}. In the first graph, we fix the value of $\kappa_1=100,000$, set $\kappa_3=10$, and display the NRPs for varying values of $\kappa_2$. As we have kept $k-m_W=4$, the same as in Figure \ref{fig:NRPmw2}, we find identical NRPs of $\phi_{c_{p-1}}$ in the first graph of Figure \ref{fig:NRPmw3} and the left panel of Figure \ref{fig:NRPmw2}.

This follows from a straightforward generalisation of Proposition \ref{Prop infty}, implying that for $s=2,3,\ldots,p-1$ the joint distribution
of $\left(\hat\kappa_{p-s+1},\ldots,\hat{\kappa}_{p}\right)$ when the eigenvalues
of $\mathcal{M}^{\prime}\mathcal{M}$ are given by $\left\{ \kappa_{1}=\ldots=\kappa_{p-s}=\infty,\kappa_{p-s+1},\ldots,\kappa_{p-1},\kappa_{p}=0\right\}$
is the same as that of $\left(\hat{\kappa}_{1}^{*},\ldots,\hat{\kappa}_{s}^{*}\right)$,
which are the eigenvalues of $\Xi^{*\prime}\Xi^{*}\sim\mathcal{W}_{s}\left(k-p+s,I_{s},\mathcal{M^{*}}^{\prime}\mathcal{M}^{*}\right)$,
with the eigenvalues of $\mathcal{M^{*}}^{\prime}\mathcal{M}^{*}$
equal to $\kappa_{j}^{*}=\kappa_{p-s+j}$, for $j=1,\ldots,s$.

Consequently, this result implies that for all values of $m_W$, starting from a situation where the $p-2$ largest eigenvalues are all equal to $\infty$, decreasing the value of $\kappa_{p-2}$ alone leads to NRPs identical to those observed in simulations for $m_W=2$. This holds across all tested $\alpha\in\{ 0.10,0.05,0.01\}$ and $k-m_W\in\{1,\ldots,20\}$, as displayed in Figure \ref{fig:NRPmw2k1k2a5} and the Supplementary Material, with the NRPs consistently below their respective $\alpha$ levels.

0As Figures \ref{fig:NRPmw2} and \ref{fig:NRPmw3} further demonstrate, and our calculations and simulations of the conditional cdfs indicate, reducing any other of the $p-3$ largest eigenvalues in value does not lead to an increase in NRPs. We illustrate this further in Figure \ref{fig:NRPmw4k310} which presents the NRPs for $\phi_{c_{p-1}}$ for $m_W=3$, $\kappa_3=10$, for various values of $\kappa_1$ and $\kappa_2$ at the 5\% level. It is again clear that for smaller values of $\kappa_1$ and $\kappa_2$ the NRPs are lower than for the larger values of these eigenvalues.

The intuition for these results is clear. Lower values of $\kappa_1,\ldots,\kappa_{p-2}$ reduce the identification strength of the nuisance parameters $\gamma$, which leads to reduced power and hence lower NRPs of $\phi_{c_{p-1}}$.

\begin{figure}[H]
\begin{centering}
\includegraphics[scale=0.40]{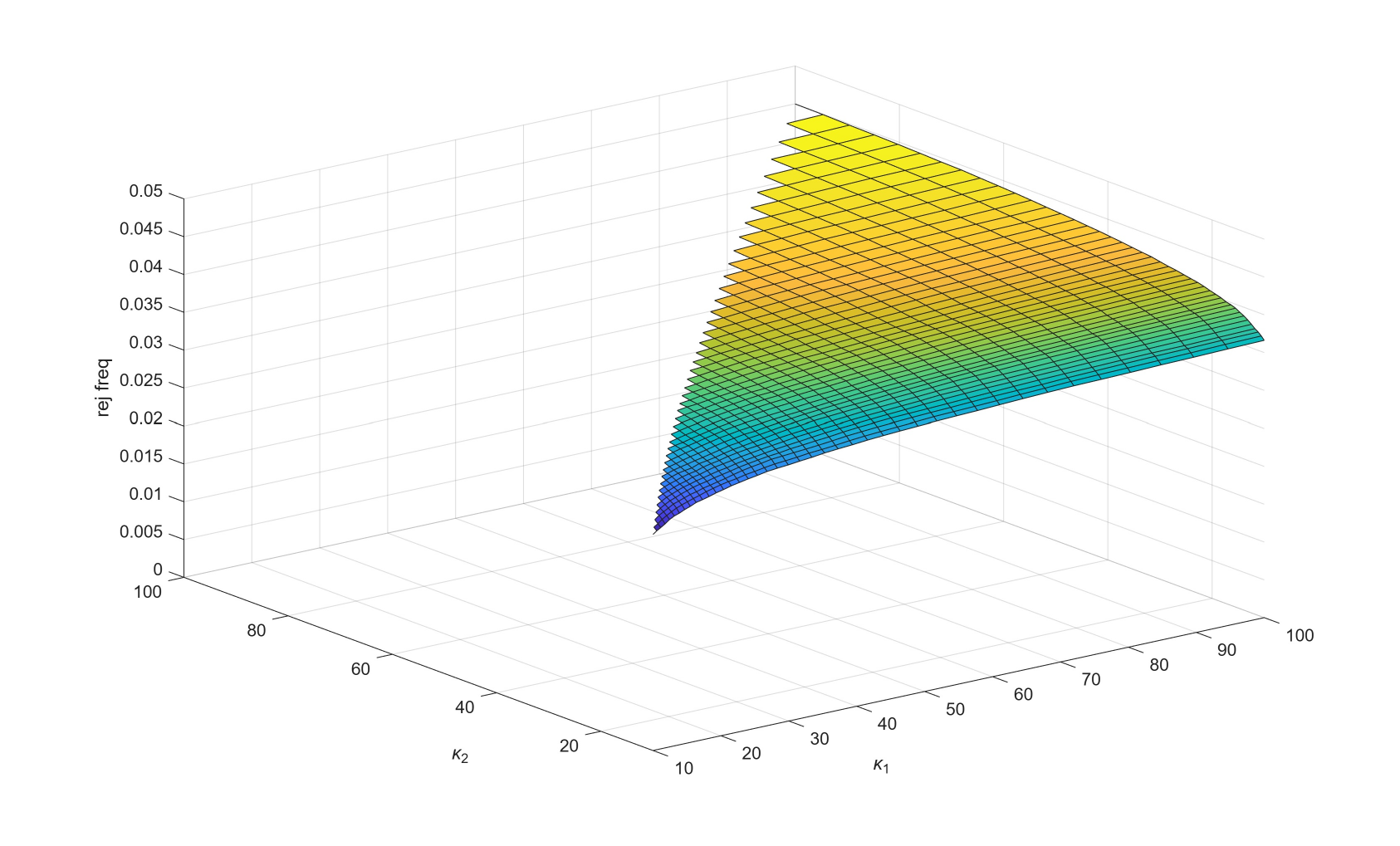}
\par\end{centering}
\caption{\label{fig:NRPmw4k310} NRPs at 5\% level of $\phi_{c_{p-1}}$ for $m_{W}=3$ and $k-m_{W}=4$, for a fixed value of $\kappa_{3}=10$, for various values of $\kappa_{1}$ and $\kappa_{2}$. See further notes to Figure \ref{fig:NRPlarge} and footnote \ref{fnotekap}.}
\end{figure}

We therefore find by direct calculation of the conditional cdfs and extensive simulations that the $\phi_{c_{p-1}}$ test
controls size also for the cases where all or some of the $\kappa_{j}$,
$j=1,\ldots,p-2$ are less then $\infty$, with the main result
stated in the following theorem.

\begin{thm} \label{ResSize} Consider, under the null $H_{0}:\beta=\beta_{0}$
and for given values of $\kappa_{p-1}$ and $\kappa_{p}=0$, the rejection
probability of the subset $AR$ test using the critical values based
on conditioning on the second-smallest eigenvalue. The rejection probability
when $\kappa_{1:(p-2)}<\infty$ is bounded from above by the rejection
probability when $\kappa_{1:(p-2)}=\infty$,

\[
\mathbb{P}\left(AR_{n}\left(\beta_{0}\right)>c_{1-\alpha}\left(\hat{\kappa}_{p-1},k-m_{W}\right)\right)_{\kappa_{1:(p-2)}<\infty,\kappa_{p-1},\kappa_{p}=0}
\]
\begin{equation}
\leq\mathbb{P}\left(AR_{n}\left(\beta_{0}\right)>c_{1-\alpha}\left(\hat{\kappa}_{p-1},k-m_{W}\right)\right)_{\kappa_{1:(p-2)}=\infty,\kappa_{p-1},\kappa_{p}=0}\leq\alpha\label{eq:ARres}. 
\end{equation}
$\phi_{c_{p-1}}$ has therefore correct size for $k-m_W\in \{1,\ldots,20\}$ and $\alpha \in \{0.10,0.05,0.01\}$, as the correct size claim for $\phi_{c_1}$ has been verified by GKM for these constellations.

\end{thm}
\begin{proof}
By direct calculation and simulation, we found for various examples with different values of $p$, $k$, $\kappa_1,\ldots,\kappa_{p-1}$, $a$ and $r$ that the conditional cdf of $\hat\kappa_p|\hat\kappa_{p-1}$
\[
\mathbb{P}\left(\hat\kappa_p\leq a|\hat\kappa_{p-1}=r;\kappa_1,\ldots\kappa_{p-1},\kappa_p=0 \right),
\]
was monotonically decreasing in increasing values of any of $\kappa_1,\ldots\,\kappa_{p-2}$. As $\phi_{c_{p-1}}$ has correct size when $\kappa_1=\ldots=\kappa_{p-2}=\infty$, it follows from this observation that the test has NRP less than size when the values of these eigenvalues are less than $\infty$. We found this to be the case by simulation of the NRPs of the test for $m_W=2$, $k-m_W\in\{1,\ldots,20\}$ and $\alpha\in\{0.10,0.05,0.01\}$, using a range of values of $\kappa_1$ and $\kappa_2$ as detailed in footnote \ref{fnotekap}. From a generalisation of Proposition \ref{Prop infty}, it follows that the test behaviour is the same for larger values of $m_W$ when only  $\kappa_{p-2}$ is reduced in value from $\infty$. As found from further simulations of the NRPs of the test for $m_W=3$ and the calculations and simulations of the conditional cdfs, reducing any further eigenvalues in value from $\infty$ does not lead to an increase in NRPs, hence the theorem holds.
\end{proof}

\vspace{10pt}
We can see the combination of Proposition \ref{Prop infty} and Theorem
\ref{ResSize} as an extension of Theorem 1 in GKM, which states for
the $p=2$ case that the distribution of the minimum eigenvalue $\hat{\kappa}_{2}$
is monotonically decreasing in $\kappa_{1}$, and converging to $\chi_{k-1}^{2}$
as $\kappa_{1}\rightarrow\infty$.

For the power of the test and the result that $\phi_{c_{p-1}}$ is the best feasible test within the class of tests that use a conditional critical value function, conditioning on a single eigenvalue, we state the next theorem,

\begin{thm}
\label{thm:power}
 The test $\phi_{c_{p-1}}$ has strictly higher power than the GKM test $\phi_{c_1}$ when $m_W>1$. $\phi_{c_{p-1}}$ is the best feasible test when $m_W>1$ in the class of tests based on a conditional critical value function, conditioning on a single eigenvalue.
\end{thm}
\begin{proof}
As $c_{1-\alpha}\left(\hat{\kappa}_{p-1},k-m_W\right)<c_{1-\alpha}\left(\hat\kappa_{1},k-m_W\right)$ when $m_W>1$, it follows that $\phi_{c_{p-1}}$ has strictly higher power than $\phi_{c_1}$ when $m_W>1$. As $c_{1-\alpha}\left(\hat{\kappa}_{p-1},k-m_W\right)<c_{1-\alpha}\left(\hat\kappa_{p-2},k-m_W\right)$ when $m_W>1$, it follows that $\phi_{c_{p-1}}$ is the best feasible test. The second-smallest eigenvalue $\hat{\kappa}_{p-1}$ is the best single-eigenvalue indicator of identification strength of the nuisance parameters $\gamma$.
\end{proof}

Again, note that the statements in Theorem \ref{thm:power} hold specifically for the $\alpha\in\{0.10,0.05,0.01\}$ and $k-m_W\in\{1,\ldots,20\}$ values used for the tabulation of the critical value functions in GKM.

\section{Feasible \texorpdfstring{$AR$}{AR} Test Asymptotic Power Comparisons}
\label{sec:Sims}

In Section 3, GKM specify the mild conditions under which their finite-sample results extend to asymptotic results when the instruments are random, the reduced-form and first-stage errors are not necessarily normally distributed, and $\Omega$ is unknown. Instead they assume that the random vectors $\left(\varepsilon_i,Z_i',V_{Xi}',V_{Wi}'\right)$ for $i=1,\ldots,n$ are i.i.d.\ with distribution $F$.

Let $U_i=\left(\varepsilon_i+V_{Wi}'\gamma,V_{Wi}'\right)'$. Then GKM consider the parameter space $\mathcal{F}$ for $\left(\gamma,\Pi_W,\Pi_X,F\right)$ under the null hypothesis $H_0:\beta=\beta_0$,
\begin{align}
\label{eq:parsF}
    \mathcal{F} &=\{(\gamma, \Pi_W, \Pi_X, F): \gamma\in \mathbb{R}^{m_W}, \Pi_W\in\mathbb{R}^{k \times m_W}, \Pi_X \in \mathbb{R}^{k\times m_X}, \notag  \\
    &\mathbb{E}_F(||T_i||^{2+\delta_1})\leq B, \text{ for } T_i \in \{\text{vec}(Z_iU_i'), U_i, Z_i\}, \notag \\
    &\mathbb{E}_F[Z_iV_i']=0^{k \times (m+1)}, \mathbb{E}_F\left(\text{vec}(Z_iU_i')(\text{vec}(Z_iU_i'))'\right) = \mathbb{E}_F\left(U_iU_i'\right) \otimes \mathbb{E}_F \left(Z_iZ_i'\right), \notag \\
    &\kappa_{min}(A)\geq\delta_2 \text{ for } A \in \{ \mathbb{E}_F(Z_iZ_i'), \mathbb{E}_F(U_iU_i') \}\}, 
\end{align}
for some $\delta_1>0$, $\delta_2>0$ and $B<\infty$.

The feasible subset $AR$ test is given by
\begin{align*}
AR_{n,f}\left(\beta_{0}\right) & =\frac{\left(y_{0}-W\hat{\tilde{\gamma}}_{L}\right)^{\prime}P_{Z}\left(y_{0}-W\hat{\tilde{\gamma}}_{L}\right)}{\left(1,-\hat{\tilde{\gamma}}_{L}^{\prime}\right)\hat{\Omega}_0\left(1,-\hat{\tilde{\gamma}}_{L}^{\prime}\right)^{\prime}}\\
 & =\min\,\text{eval}\left(\left(y_{0},W\right)^{\prime}P_{Z}\left(y_{0},W\right)\hat{\Omega}_0^{-1}\right),
\end{align*}
where
\[
\hat{\Omega}_0=\left(y_{0},W\right)^{\prime}M_{Z}\left(y_{0},W\right)/(n-k-1),
\]
and $\hat{\tilde{\gamma}}_{L}$ is the LIML estimator of $\gamma$
in restricted model (\ref{eq:restrmod}). The additional degree of
freedom loss is because we include a constant in the model by taking
all variables in deviations from sample means. Note that the feasible
test is identical to the LIML-based \citet{Basmann1960} test for overidentifying
restrictions in the restricted model.

As in GKM, denote the $p$ eigenvalues of the matrix $\left(y_{0},W\right)^{\prime}P_{Z}\left(y_{0},W\right)\hat\Omega_0^{-1}$
by $\hat{\kappa}_{1n}\geq\hat{\kappa}_{2n}\geq\ldots\geq\hat{\kappa}_{pn}$,
so that $AR_{n,f}\left(\beta_{0}\right)=\hat{\kappa}_{pn}$. GKM (Theorem 5, p 502) then show the result that for the parameter space $\mathcal{F}$ defined in (\ref{eq:parsF}) the feasible conditional subvector $AR$ test rejects $H_0$ at nominal size $\alpha$ asymptotically if
$$
AR_{n,f}\left(\beta_0\right)>c_{1-\alpha}\left(\hat\kappa_{1n},k-m_W\right),
$$
where $c_{1-\alpha}(.,.)$ is the same conditional critical value function as for the finite sample results in Section \ref{sec:finite_sample_analysis}. Again, these critical value functions are tabulated in GKM for $\alpha\in\{0.10,0.05,0.01\}$ and $k-m_W\in\{1,\ldots,20\}$, using a grid of conditioning values, with the conditional critical values for any value of $\hat\kappa_{1n}$ obtained by linear interpolation.

From the proof of GKM's Theorem 5 (GKM, Appendix A.4), it follows that the feasible conditional subvector $AR$ test also rejects $H_0$ at nominal size $\alpha$ asymptotically if we condition on the second-smallest eigenvalue, that is, if
\[
AR_{n,f}\left(\beta_0\right)>c_{1-\alpha}\left(\hat\kappa_{(p-1)n},k-m_W\right).
\]

This can be seen from GKM (equation (A.44), p 521), which can be written in our case as
\begin{align}
\mathbb{P}_{F_n} & \left(AR_{n,f}(\beta_0)>c_{1-\alpha}(\hat\kappa_{(p-1)n},k-m_W)\right) \nonumber \\ 
& \rightarrow \mathbb{P}\left(\kappa(p)>c_{1-\alpha}(\kappa(p-1),k-m_W)\right) \leq \alpha, \label{eq:asyres}
\end{align}
for $n\rightarrow\infty$ and under a sequence $\left\{F_n:n \geq 1 \right\}$, specified in GKM (p 515). Here $\kappa(j)$ is the $j^{th}$ ordered eigenvalue, ordered in non-increasing order, of $\tilde{W}'\tilde{W}$ , with $\tilde{W}=\left(\tilde{w}_1,\ldots,\tilde{w}_p\right)$, where $\tilde{w}_j\in\mathbb{R}^k$ for $j=1,\ldots,p$ are independent $\mathcal{N}\left(m_j,I_k\right)$ with $m_j=\left(0^{j-1 \prime},h_{1,j},0^{k-j \prime}\right)'$ for $j<p$ and $m_p=0^k$. For the definition of $h_{1,j}\geq 0$, see GKM (equation (A.32), p 515). We can therefore use the finite sample results of Proposition \ref{Prop infty} and Theorem \ref{ResSize} in Section \ref{sec:Cond2nd} with the roles of $\Xi$ and $\mathcal{M}$ played by $\tilde{W}$ and $M=\left(m_1,\ldots,m_p\right)$ respectively, from which the asymptotic size control as displayed in (\ref{eq:asyres}) follows.

For $m_W>1$, as the critical value conditioning on the second
smallest eigenvalue is smaller than that conditioning on the largest
eigenvalue, this $AR$ subvector testing procedure has strictly higher
power than the GKM test. The following simulations confirm this.

For the full model specification (\ref{eq:fullmod}), we set $m_{X}=1$,
$m_{W}=3$, $k=7$, $\beta_0=0$, $\gamma=\left(-1,1,1\right)^{\prime}$,
$n=250$, $Z_{i}\overset{i.i.d.}{\sim}~\mathcal{N}\left(0,I_{k}\right)$ and
$\left(\varepsilon_{i},V_{Xi},V_{Wi}^{\prime}\right)^{\prime}\overset{i.i.d.}{\sim}~\mathcal{N}\left(0,\Sigma\right)$.
The specific parameter values of $\Sigma$, $\pi_{x}$ and $\Pi_{W}$
are given in Appendix \ref{subsec:Monte-Carlo-Parameter}. The resulting
expectation of the concentration matrix is given by
\[
\mathbb{E}_{Z}\left(\Theta_{W}^{\prime}\Theta_{W}\right)=n\Sigma_{V_{W}V_{W}.\varepsilon}^{-1/2}\Pi_{W}^{\prime}\Pi_{W}\Sigma_{V_{W}V_{W}.\varepsilon}^{-1/2}=\left(\begin{array}{ccc}
\kappa_{1} & 0 & 0\\
0 & \kappa_{2} & 0\\
0 & 0 & \kappa_{3}
\end{array}\right).
\]
Figure \ref{fig:Power} presents power curves for $\phi_{\chi^{2}}$,
$\phi_{c_1}$ and $\phi_{c_{p-1}}$, for testing $H_{0}:\beta=0$ at level
$\alpha=0.05$, for 100,000 replications at each value of $\beta=-2,-1.95,\ldots,2$.
We consider the weakly identified situations for $\gamma$ with $\kappa=\left(\kappa_{1},\kappa_{2},\kappa_{3}\right)=\left(35,25,15\right)$
and $\kappa=\left(100,30,15\right)$ and a strongly identified case
with $\kappa=\left(100,95,90\right)$. For $\beta$ we have that the marginal concentration parameter $\mu_x^2=E(\pi_x^\prime Z^\prime Z\pi_x/\sigma_{v_x}^2) = 16$. The $\phi_{c_{p-1}}$ test controls size and clearly
dominates power in the weak-identification settings, with the tests behaving
similarly in the strongly identified setting, as expected.

For comparison, Figure \ref{fig:Power} also includes the power curves of the subset-CLR test of \cite{Kleibergen2021}, who derived the conditional critical value function under which the subset LR-test is size-correct under weak identification and efficient under strong identification. \cite{Kleibergen2021} finds that the CLR test outperforms the $\phi_{c_1}$ test in terms of power "..when the structural parameters are reasonably identified and has slightly less power when identification is weak." \citep[p 78]{Kleibergen2021}. In the weakly identified setting, we therefore expect $\phi_{c_{p-1}}$ to dominate the CLR test in power, as the top two panels of Figure \ref{fig:Power} confirm. The bottom panel with strong identification of $\gamma$ shows only marginal differences in power between the tests, with the CLR test slightly dominating for negative values of $\beta$ and $\phi_{c_{p-1}}$ dominating for larger positive values of $\beta$.

\begin{figure}[H]
\begin{centering}
\includegraphics[scale=0.65]{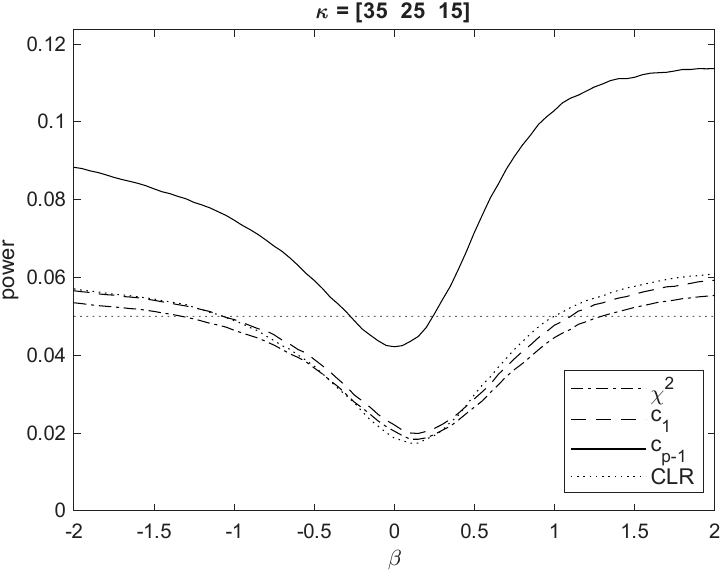}\includegraphics[scale=0.65]{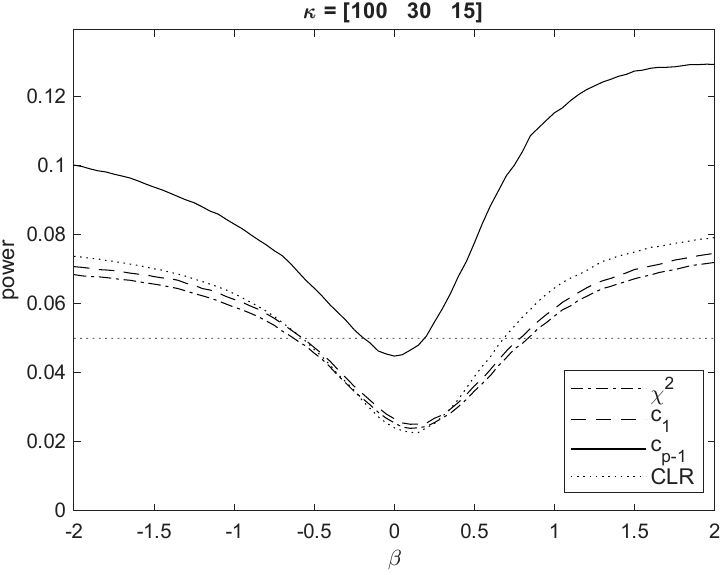}
\par\end{centering}
\begin{centering}
\includegraphics[scale=0.65]{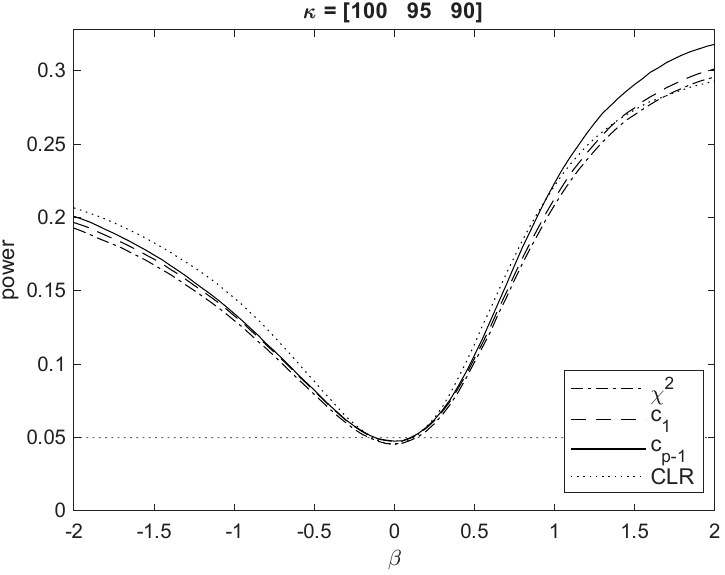}
\par\end{centering}
\caption{\label{fig:Power} Power of $\phi_{\chi^{2}}$, $\phi_{c_1}$, $\phi_{c_{p-1}}$ and CLR test of \cite{Kleibergen2021}
for testing $H_{0}:\beta=0$ at level $\alpha=0.05$. $m_{W}=3$,
$k=7$, $n=250$, $\mu_x^2=16$, for different values of $\kappa=\left(\kappa_{1},\kappa_{2},\kappa_{3}\right)$.}
\end{figure}

In Appendix \ref{subsec:Monte-Carlo-Parameter}, we further show in Figure \ref{fig:PowerX} the power curves after increasing the identification strength for $\beta$ by multiplying $\pi_x$ by 2, resulting in an increased marginal concentration parameter $\mu_x^2=64$. Only when both $\gamma$ and $\beta$ are strongly identified does the CLR test dominate the $\phi_{c_{p-1}}$ test in power, and mainly for negative values of $\beta$. For the situation with $\kappa=\left(35,25,15\right)$, the $\phi_{c_{p-1}}$ is more powerful across all values of $\beta$, whereas for the  $\kappa=\left(100,30,15\right)$ situation, $\phi_{c_{p-1}}$ is more powerful than the CLR test, apart from larger negative values of $\beta$, where the power of the two tests is very similar. It is clear that for our case of interest, a weakly-identified $\gamma$, the $\phi_{c_{p-1}}$ test can be substantially more powerful than the CLR test.

Figure \ref{fig:PowerT} in Appendix \ref{subsec:Monte-Carlo-Parameter} further shows the power graphs for the tests in the same setting as for Figure \ref{fig:Power}, but instead of the multivariate normal distribution of the errors, their distribution is multivariate $t$, with 4 degrees of freedom. These power graphs are virtually identical to those shown in Figure \ref{fig:Power}.

\subsection{Power bound comparison}
We consider the GKM-adjusted least favorable distribution (ALFD) method of \cite[EMW]{elliott2015nearly}, following GKM's approach for $m_W=1$. We include details of this algorithmic approach in Appendix \ref{sec:PowerBound}, but note that it directly extends GKM from $m_W=1$ to $m_W=2$. 

As shown in \cite[Lemma 1]{elliott2015nearly}, the calibrated EMW upper bound remains valid for any chosen null mixture. In other words, choices we make in relation to the calculation of this upper bound (i.e., null discretization, selection, and calculation) affect the convergence of the null mixture weights, hence tightness, not upper-bound validity. As our argument relates to the proximity of our power curve to this bound, we mainly care about validity---a small enough gap to a valid upper bound already supports our conclusion; a tighter bound would only make the argument stronger. 

We calculate the power bound for various $\kappa$ configurations ($(\kappa_1,\kappa_2) \in \{(100,15), \allowbreak(50,15), \allowbreak(15,15)\}$), as in this $m_W=2$ extension of GKM, more points would make the calculations overwhelming. Similarly to GKM, we have to make decisions in order to make calculations computationally tractable. Nevertheless, by restricting to $m_W=1$, GKM misses the higher-dimensional nuisance setting in which a significant improvement over their test is possible. By extending to this higher-dimensional setting, we make this weakness apparent.

More specifically, GKM compare the power bound with their test in the $m_W=1$ case, highlighting the test's near-optimality. Since our test is equivalent to theirs for $m_W=1$, we trivially inherit this. Restricting the analysis to $m_W=1$ does not provide the full picture of their test. Figure \ref{fig:Bound} highlights this by showing the power bound with the various power curves for $m_W=2$ and $(\kappa_1,\kappa_2) \in \{(100,15), (50,15), (15,15)\}$. The figure shows again the  suboptimality of $\phi_{c_1}$ and the significant power gains of $\phi_{c_{p-1}}$. The power curve of $\phi_{c_{p-1}}$ is close to the bound in the setting where $\phi_{c_{1}}$ improves little on the standard $\phi_{\chi^2}$ test (a large largest eigenvalue, $\kappa_1=100$), and remains close to the bound when we significantly decrease the largest eigenvalue ($\kappa_1=50$).

\begin{figure}[H]
\begin{centering}
\includegraphics[scale=0.65]{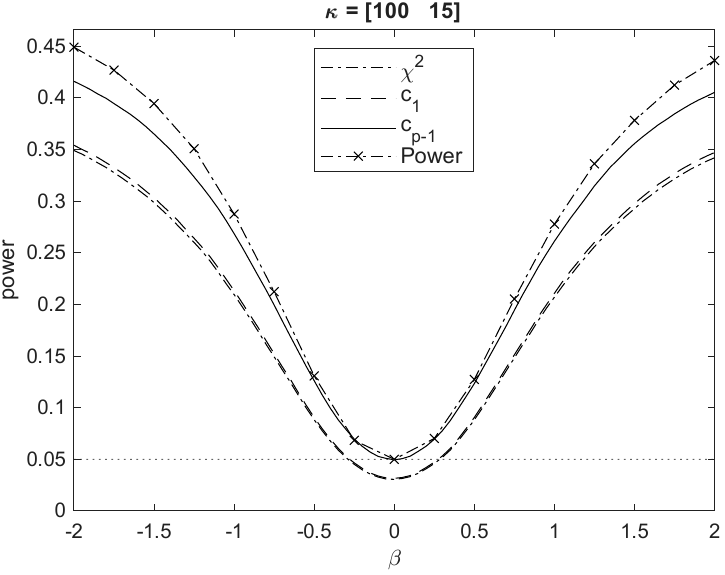}\includegraphics[scale=0.65]{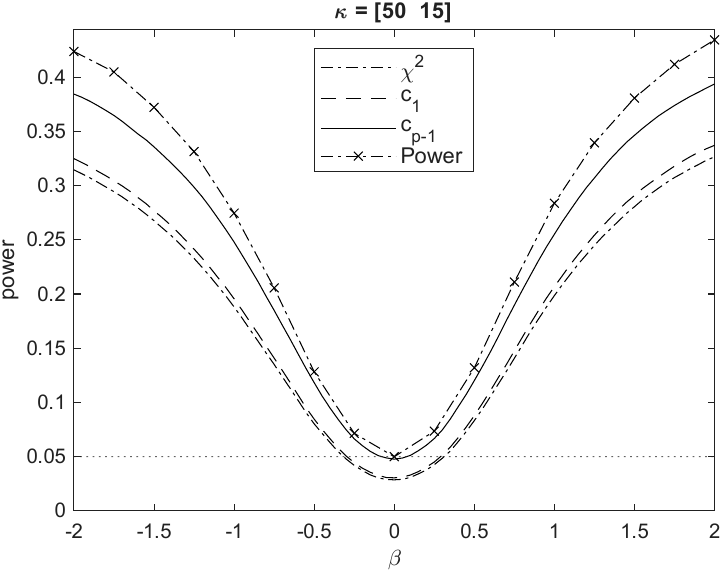}
\par\end{centering}
\begin{centering}
\includegraphics[scale=0.65]{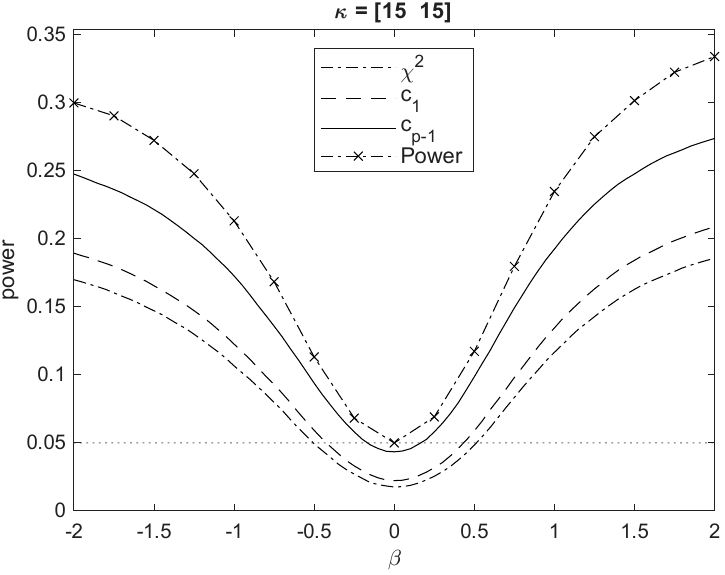}
\par\end{centering}
\caption{\label{fig:Bound} Power of $\phi_{\chi^{2}}$, $\phi_{c_1}$, $\phi_{c_{p-1}}$ and a valid power bound
for testing $H_{0}:\beta=0$ at level $\alpha=0.05$. $m_{W}=2$,
$k=7$, $n=250$, for different values of $\kappa=\left(\kappa_{1},\kappa_{2}\right)$.}
\end{figure}

We include the setting with $(\kappa_1, \kappa_2)=(15,15)$, where in principle our test adds little additional information beyond what GKM capture. Nevertheless, even then, the power curve of $\phi_{c_{p-1}}$ lies closer to the bound than to that of $\phi_{c_{1}}$.

Consequently, we term the $\phi_{c_{p-1}}$ test "best feasible", as it strikes the optimal balance between computational simplicity and statistical power.  We define feasibility as the class of tests that restrict conditioning to a single component of the spectral summary of identification strength. Restricting our attention to such tests allows us to use a highly tractable approximate conditional density that relies solely on standard scalar-input functions.

Within this computationally simple class, our approach represents the best test. By isolating the second-smallest eigenvalue, we condition on the specific margin that is most informative regarding the strength of identification. Although we show only limited potential for improvements in this low-dimensional setting, these gains could be more meaningful as dimensions grow large. Therefore, we emphasize that future work should explore ways to make higher-order conditioning sets feasible and further improve power.

\section{Empirical Analysis}
\label{sec:empirical}
\citet*{tanaka2010risk} (TCN), use experimental data collected from Vietnamese villages and are the data used in the empirical section of GKM. TCN estimate linear IV models to study determinants of risk and time preferences. The dependent variable in their models is the curvature of the utility function, denoted by $\sigma$ in their notation. They report two specifications from which we will use the second, where they decompose household income into mean village income (\textit{Mean Income}) and relative income within the village (\textit{Relative Income}), resulting in two endogenous regressors. Their sample is random by design, and TCN assume homoskedasticity.  

To showcase the improvement of $\phi_{c_{p-1}}$  over $\phi_{c_1}$ we require $m_W>1$. Hence, we focus here on the model with mean income and relative income as the two endogenous variables apply subvector inference on the coefficients of each of the exogenous variables in turn. Note that testing the coefficients on either endogenous variable reduces to an $m_W=1$ setting, where both tests are equivalent.

The exogenous variables included in the model are \textit{Chinese, Age, Gender, Education, Distance to market}, and \textit{South}. The two excluded exogenous variables used as instruments are \textit{Rainfall}, and \textit{"Head of household can't work"}.  An exogenous variable under test is treated as an additional instrumental variable, so we then have $k-m_W=1$.

The size-$\alpha$ conditional subvector AR tests $\phi_{\chi^2}$, $\phi_{c_1}$ and $\phi_{c_{p-1}}$ for testing $H_0:\beta=\beta_0$ can be performed using the same steps as in GKM (Algorithm 1, p 506). $(1-\alpha)$-level confidence sets for $\beta$ can be obtained by grid search over a sufficiently large range of values for $\beta_0$, as illustrated in Figure \ref{fig:subar_stat_plus_cv}. The graph in Figure \ref{fig:subar_stat_plus_cv} plots the subvector AR statistic for the \textit{South} coefficient, together with both the conditional and unconditional 5\% critical values. Note that the conditional critical values vary with $\beta_0$. Table \ref{tab:vietnam_results_empirical} reports the resulting 95\% confidence sets, also for the other coefficients in the model.

\begin{figure}
    \centering
    \includegraphics[scale=0.40]{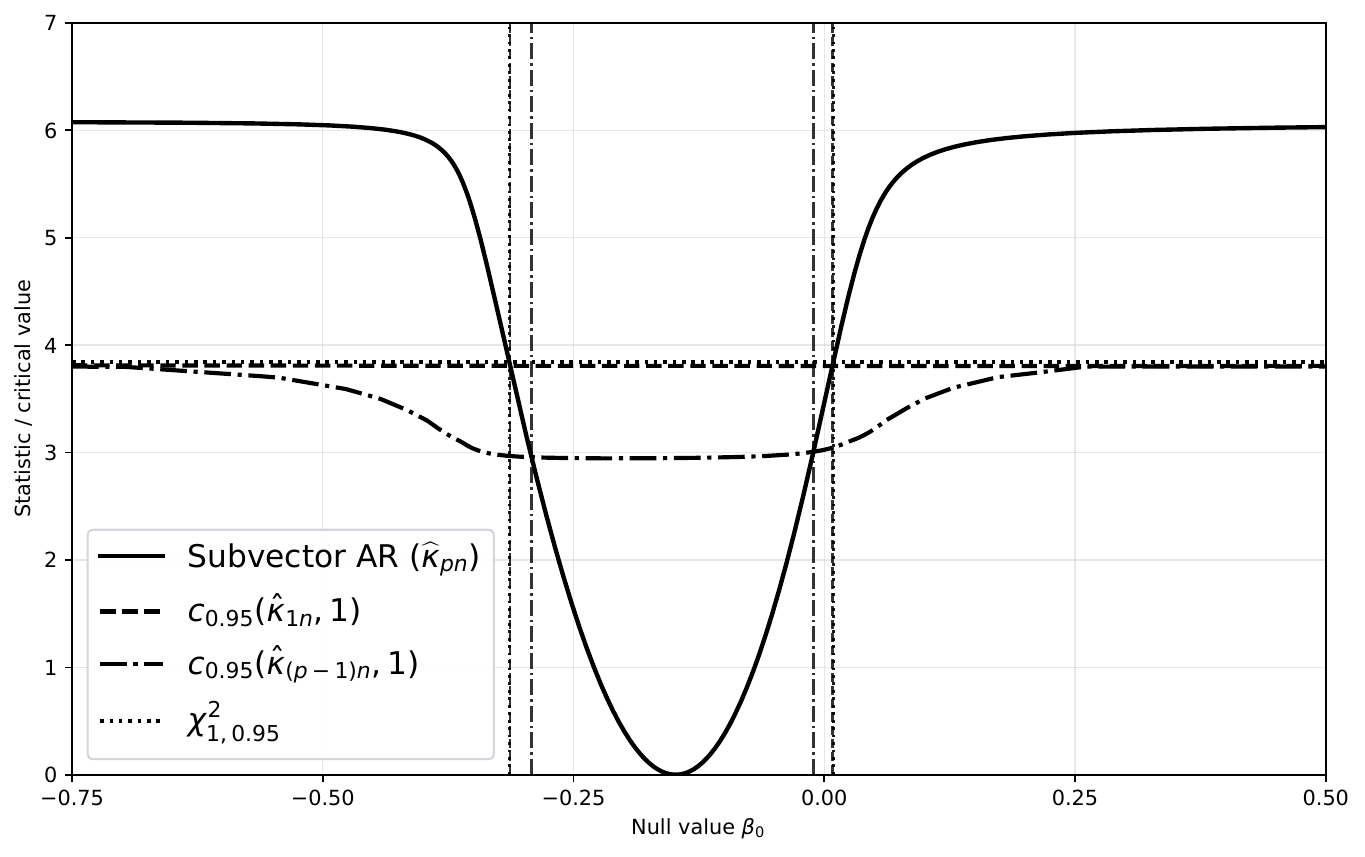}
    \caption{Subvector AR statistic for \textit{South} of Table \ref{tab:vietnam_results_empirical}, with conditional and unconditional critical values. Vertical lines indicate 95\% confidence sets reported in the same table.}
    \label{fig:subar_stat_plus_cv}
\end{figure}

The $\phi_{c_{p-1}}$ based confidence sets are uniformly smaller across the exogenous variables than those based on $\phi_{c_1}$. For the (non-rounded) confidence sets we get a length reduction of around 12\% for most variables, compared to less than 1\% length reductions for the $\phi_{c_1}$ confidence sets in comparison to those based on $\phi_{\chi^2}$. Furthermore, for the exogenous variable \textit{South} the $\phi_{c_1}$ based  confidence set includes zero, whereas the $\phi_{c_{p-1}}$ based confidence set excludes zero. This is further highlighted in Figure \ref{fig:subar_stat_plus_cv}. Our test $\phi_{c_{p-1}}$ not only strengthens the results reported in GKM, but also affects some of its conclusions, indicating the practical importance of our proposed test procedure.

\begin{table}[!htbp]
\centering
\footnotesize
\begin{threeparttable}
\caption{Replication of TCN (Table 5) and GKM (Table 2, Specification 2). 2SLS estimates and  95\% confidence
sets based on $\phi_{\chi^2}$, $\phi_{c_1}$ and $\phi_{c_{p-1}}$.}
\label{tab:vietnam_results_empirical}

\begin{tabular}{@{}lcl@{\,: }lc@{}}
\toprule
& \multicolumn{3}{c}{Dependent variable} \\
\cmidrule(l){2-4}
& \multicolumn{3}{c}{$\sigma$ (value function curvature)} \\
\cmidrule(l){2-4}
Variable & 2SLS estimate & \multicolumn{2}{l}{95\% confidence set} \\
\midrule

Chinese & $-0.096$ & $\chi^2$ & $[-0.396,\ 0.166]$
& \multirow{2}{*}{\reductionarrow{$-0.58\%$}} \\
        &          & $c_1$ & $[-0.394,\ 0.165]$ & \multirow{2}{*}{\reductionarrow{$-12.2\%$}}\\
        &          & $c_{p-1}$  & $[-0.357,\ 0.134]$ & \\
\addlinespace

Age & $-0.006$ & $\chi^2$ & $[-0.011,\ -0.002]$  & \multirow{2}{*}{\reductionarrow{$-0.59\%$}}\\
    &          & $c_1$ & $[-0.011,\ -0.002]$  & \multirow{2}{*}{\reductionarrow{$-12.1\%$}}\\
    &          & $c_{p-1}$  & $[-0.011,\ -0.002]$ \\
\addlinespace

Gender & $-0.006$ & $\chi^2$ & $[-0.121,\ 0.118]$  & \multirow{2}{*}{\reductionarrow{$-0.58\%$}}\\
       &          & $c_1$ & $[-0.120,\ 0.117]$  & \multirow{2}{*}{\reductionarrow{$-12.2\%$}}\\
       &          & $c_{p-1}$  & $[-0.107,\ 0.103]$ \\
\addlinespace

Edu & $-0.028$ & $\chi^2$ & $[-0.055,\ -0.008]$  & \multirow{2}{*}{\reductionarrow{$-0.99\%$}}\\
    &          & $c_1$ & $[-0.055,\ -0.008]$  & \multirow{2}{*}{\reductionarrow{$-5.94\%$}}\\
    &          & $c_{p-1}$  & $[-0.054,\ -0.009]$ \\
\addlinespace

Market & $-0.013$ & $\chi^2$ & $[-0.045,\ 0.016]$  & \multirow{2}{*}{\reductionarrow{$-0.59\%$}}\\
       &          & $c_1$ & $[-0.044,\ 0.016]$  & \multirow{2}{*}{\reductionarrow{$-12.1\%$}}\\
       &          & $c_{p-1}$  & $[-0.041,\ 0.013]$ \\
\addlinespace

South & $-0.148$ & $\chi^2$ & $[-0.314,\ 0.008]$  & \multirow{2}{*}{\reductionarrow{$-0.52\%$}}\\
      &          & $c_1$ & $[-0.313,\ 0.008]$  & \multirow{2}{*}{\reductionarrow{$-12.6\%$}}\\
      &          & $c_{p-1}$  & $[-0.292,\ -0.011]$ \\
\addlinespace

Constant & $0.992$ & $\chi^2$ & $[0.675,\ 1.366]$  & \multirow{2}{*}{\reductionarrow{$-0.75\%$}}\\
         &         & $c_1$ & $[0.676,\ 1.363]$  & \multirow{2}{*}{\reductionarrow{$-10.4\%$}}\\
         &         & $c_{p-1}$  & $[0.711,\ 1.327]$ \\
\addlinespace

Relative Income & $0.049$ & $\chi^2$ & $[-0.339,\ 0.638]$  \\
                     &         & $c_1$ & $[-0.334,\ 0.628]$  \\
                     &         & $c_{p-1}$  & $[-0.334,\ 0.628]$ \\
\addlinespace

Mean Income  & $0.010$ & $\chi^2$ & $[-0.001,\ 0.022]$  \\
                 &         & $c_1$ & $[0.001,\ 0.020]$  \\
                 &         & $c_{p-1}$  & $[0.001,\ 0.020]$ \\
\bottomrule
\end{tabular}

\begin{tablenotes}[flushleft]
\footnotesize
\item Sample size is 181. For Relative Income and Mean Income, $m_W=1$, so that the confidence sets based on $\phi_{c_1}$ and $\phi_{c_{p-1}}$ are the same. The indicated percentage differences relate to the changing lengths of the respective confidence sets.
\end{tablenotes}
\end{threeparttable}
\end{table}

\section{Concluding Remarks}
\label{seq:Conclusions}

We have shown that the \cite*{GuggenbergerKleiMavrQE2019} conditional critical value function for the subvector $AR$ test in the linear model under homoskedasticity applies when conditioning on the second-smallest eigenvalue instead of the largest eigenvalue of a non-central Wishart-distributed matrix. This test procedure, which conditions on the second-smallest eigenvalue, controls size and has power strictly higher than the GKM test. 
\cite*{Guggenberger_Kleibergen_Mavroeidis_2024} show further that the same conditional critical value function applies to a situation with conditional heteroskedasticity, where the covariance matrix has an approximate Kronecker product structure. The subvector $AR$ statistic is redefined to account for this Kronecker product structure, and the eigenvalues come from the associated concentration matrix. We present further details in Appendix \ref{sec:Het}. As in the homoskedastic case, for this conditional heteroskedasticity structure, conditioning on the second-smallest eigenvalue controls size and has strictly higher power than conditioning on the largest eigenvalue.

\section*{Appendix}

\global\long\def\thesection{A}%
\global\long\def\theequation{A.\arabic{equation}}%
\global\long\def\thetable{A\arabic{table}}%
\global\long\def\thefigure{A\arabic{figure}}%
\setcounter{table}{0}\setcounter{equation}{0}\setcounter{figure}{0}

\subsection{Proof of Proposition 1}

\label{subsec:Proof-of-Proposition}

Drawing heavily from GKM, including their notation, we have the $k\times\left(m_{W}+1\right)$
matrix $\Xi\sim\mathcal{N}\left(\mathcal{M},I_{kp}\right)$, where
$\text{var}\left(\text{vec}\left(\Xi\right)\right)=I_{kp}$, with
$p=m_{W}+1$, and 
\[
\Xi'\Xi\sim\mathcal{W}_{m_{W}+1}\left(k,I_{m_{W}+1},\mathcal{M}'\mathcal{M}\right).
\]
Then partition $\Xi$ as
\[
\Xi=\left[\begin{array}{cc}
\Xi_{11} & \Xi_{12}\\
\Xi_{21} & \Xi_{22}
\end{array}\right],
\]
where $\Xi_{11}$ is a $\left(k-m_{W}+1\right)\times2$ matrix, $\Xi_{12}$
is a $\left(k-m_{W}+1\right)\times\left(m_{W}-1\right)$ matrix, $\Xi_{21}$
is a $\left(m_{W}-1\right)\times2$ matrix and $\Xi_{22}$ is a $\left(m_{W}-1\right)\times\left(m_{W}-1\right)$
matrix. Partition $\mathcal{M}$ conformably with $\Xi$ and let $\mu_{i}$,
$i=1,\ldots,m_{W}$, denote the singular values of $\mathcal{M}$, with
$\mu_{1}\geq\mu_{2}\geq\ldots\geq\mu_{m_{W}}\geq0$. Without loss
of generality, we can set
\[
\mathcal{M}=\left[\begin{array}{cc}
\mathcal{M}_{11} & 0^{\left(k-m_{W}+1\right)\times\left(m_{W}-1\right)}\\
0^{\left(m_{W}-1\right)\times2} & \mathcal{M}_{22}
\end{array}\right],
\]
where
\[
\mathcal{M}_{11}=\left[\begin{array}{cc}
0^{\left(k-m_{W}\right)\times1} & 0^{\left(k-m_{W}\right)\times1}\\
0 & \mu_{m_{W}}
\end{array}\right]
\]
\[
\mathcal{M}_{22}=\text{diag}\left(\mu_{1},\ldots,\mu_{m_{W}-1}\right).
\]

It follows that
\begin{align*}
\Xi_{11}^{\prime}\Xi_{11} & \sim\mathcal{W}_{2}\left(k-m_{W}+1,I_{2},\mathcal{M}_{11}^{\prime}\mathcal{M}_{11}\right)\\
\Xi_{22}^{\prime}\Xi_{22} & \sim\mathcal{W}_{m_{W}-1}\left(m_{W}-1,I_{m_{W}-1},\mathcal{M}_{22}^{\prime}\mathcal{M}_{22}\right).
\end{align*}

Let the $\left(m_{W}+1\right)\times\left(m_{W}+1\right)$ matrix $O$
be given by
\[
O=\left[\begin{array}{cc}
O_{11} & \Xi_{21}^{\prime}\Xi_{22}^{-1\prime}O_{22}\\
-\Xi_{22}^{-1}\Xi_{21}O_{11} & O_{22}
\end{array}\right],
\]
with
\begin{align*}
O_{11} & =\left(I_{2}+\Xi_{21}^{\prime}\Xi_{22}^{-1\prime}\Xi_{22}^{-1}\Xi_{21}\right)^{-1/2}\\
O_{22} & =\left(I_{m_{W}-1}+\Xi_{22}^{-1}\Xi_{21}\Xi_{21}^{\prime}\Xi_{22}^{-1\prime}\right)^{-1/2}.
\end{align*}
We have that $O'O=OO'=I_{m_{W}+1}$. Then define
\begin{align*}
\tilde{\Xi} & \coloneqq\Xi O\\
 & =\left[\begin{array}{cc}
\tilde{\Xi}_{11} & \tilde{\Xi}_{12}\\
0 & \tilde{\Xi}_{22}
\end{array}\right],
\end{align*}
with
\begin{align*}
\tilde{\Xi}_{11} & =\left(\Xi_{11}-\Xi_{12}\Xi_{22}^{-1}\Xi_{21}\right)O_{11}\\
\tilde{\Xi}_{12} & =\left(\Xi_{11}\Xi_{21}^{\prime}\Xi_{22}^{-1\prime}+\Xi_{12}\right)O_{22}\\
\tilde{\Xi}_{22} & =\left(\Xi_{21}\Xi_{21}^{\prime}\Xi_{22}^{-1\prime}+\Xi_{22}\right)O_{22}
\end{align*}
We have that for the eigenvalues $\hat{\kappa}_{j}$, $j=1,\ldots,m_{W}+1,$
\[
\hat{\kappa}_{j}\left(\tilde{\Xi}'\tilde{\Xi}\right)=\hat{\kappa}_{j}\left(\Xi'\Xi\right).
\]

Further,
\[
\tilde{\Xi}'\tilde{\Xi}=\left[\begin{array}{cc}
\tilde{\Xi}_{11}^{\prime}\tilde{\Xi}_{11} & \tilde{\Xi}_{11}^{\prime}\tilde{\Xi}_{12}\\
\tilde{\Xi}_{12}^{\prime}\tilde{\Xi}_{11} & \tilde{\Xi}_{12}^{\prime}\tilde{\Xi}_{12}+\tilde{\Xi}_{22}^{\prime}\tilde{\Xi}_{22}
\end{array}\right].
\]
Let
\[
\tilde{A}:=\tilde{\Xi}'\tilde{\Xi}=\left[\begin{array}{cc}
\tilde{A}_{11} & \tilde{A}_{12}\\
\tilde{A}_{12}^{\prime} & \tilde{A}_{22}
\end{array}\right],
\]
then, as above, let
\[
\tilde{D}=\left[\begin{array}{cc}
\tilde{D}_{11} & \tilde{A}_{12}\tilde{A}_{22}^{-1}\tilde{D}_{22}\\
-\tilde{A}_{22}^{-1}\tilde{A}_{12}^{\prime}\tilde{D}_{11} & \tilde{D}_{22}
\end{array}\right],
\]
with
\begin{align*}
\tilde{D}_{11} & =\left(I_{2}+\tilde{A}_{12}\tilde{A}_{22}^{-1}\tilde{A}_{22}^{-1}\tilde{A}_{12}^{\prime}\right)^{-1/2}\\
\tilde{D}_{22} & =\left(I_{m_{W}-1}+\tilde{A}_{22}^{-1}\tilde{A}_{12}^{\prime}\tilde{A}_{12}\tilde{A}_{22}^{-1}\right)^{-1/2},
\end{align*}
and so
\[
\tilde{A}\tilde{D}=\left[\begin{array}{cc}
\left(\tilde{A}_{11}-\tilde{A}_{12}\tilde{A}_{22}^{-1}\tilde{A}_{12}^{\prime}\right)\tilde{D}_{11} & \left(\tilde{A}_{11}\tilde{A}_{12}\tilde{A}_{22}^{-1}+\tilde{A}_{12}\right)\tilde{D}_{22}\\
0 & \left(\tilde{A}_{12}^{\prime}\tilde{A}_{12}\tilde{A}_{22}^{-1}+\tilde{A}_{22}\right)\tilde{D}_{22}
\end{array}\right].
\]
Define 
\[
\tilde{C}\coloneqq\tilde{D}\tilde{'A}\tilde{D}=\left[\begin{array}{cc}
\tilde{C}_{11} & \tilde{C}_{12}\\
\tilde{C}_{12}^{\prime} & \tilde{C}_{22}
\end{array}\right],
\]
where
\begin{align*}
\tilde{C}_{11} & =\tilde{D}_{11}^{\prime}\left(\tilde{A}_{11}-\tilde{A}_{12}\tilde{A}_{22}^{-1}\tilde{A}_{12}^{\prime}\right)\tilde{D}_{11}\\
\tilde{C}_{12} & =\tilde{D}_{11}^{\prime}\left(\tilde{A}_{11}-\tilde{A}_{12}\tilde{A}_{22}^{-1}\tilde{A}_{12}^{\prime}\right)\tilde{A}_{12}\tilde{A}_{22}^{-1}\tilde{D}_{22}\\
\tilde{C}_{22} & =\tilde{D}_{22}^{\prime}\tilde{A}_{22}^{-1}\tilde{A}_{12}^{\prime}\left(\tilde{A}_{11}\tilde{A}_{12}\tilde{A}_{22}^{-1}+\tilde{A}_{12}\right)\tilde{D}_{22}+\tilde{D}_{22}^{\prime}\left(\tilde{A}_{12}^{\prime}\tilde{A}_{12}\tilde{A}_{22}^{-1}+\tilde{A}_{22}\right)\tilde{D}_{22}.
\end{align*}

Consider $\mu_{m_{W}-1}\rightarrow\infty$, which implies that the
singular values $\mu_{j}\rightarrow\infty$ for $j=1,.\ldots,m_{W}-1$,
which are the diagonal elements of the diagonal matrix \textbf{$\mathcal{M}_{22}$}.
As $\Xi_{22}\sim\mathcal{N}\left(\mathcal{M}_{22},I_{\left(m_{W}-1\right)^{2}}\right)$
it follows that all diagonal elements of $\Xi_{22}$ approach infinity
as $\mu_{m_{W}-1}\rightarrow\infty$, whereas the off-diagonal elements
are i.i.d. $\mathcal{N}\left(0,1\right)$ distributed. It follows that
$\Xi_{22}^{-1}\stackrel{p}{\rightarrow}0$ as $\mu_{m_{W}-1}\rightarrow\infty$,
where convergence in probability is with respect to the singular values
$\mu_{j}\rightarrow\infty$ for $j=1,.\ldots,m_{W}-1$. We can then
write the sequence as $\left(\Xi_{22}\left(\mu_{W-1}\right)\right)^{-1}$,
with the convergence in probability then defined as
\[
\lim_{\mu_{m_{W}-1}\rightarrow\infty}\mathbb{P}\left(\Vert\left(\Xi_{22}\left(\mu_{W-1}\right)\right)^{-1}-0\Vert>\varepsilon\right)=0
\]
for any $\varepsilon>0$. We will use this concept of convergence
in probability throughout below.

As the elements $\Xi_{12}$ are also i.i.d. $\mathcal{N}$$\left(0,1\right)$
distributed, we get
\begin{align*}
O_{11} & \stackrel{p}{\rightarrow}I_{2};\,\,O_{22}\stackrel{p}{\rightarrow}I_{m_{W}-1},
\end{align*}
from which it follows that 
\[
\tilde{\Xi}_{11}\stackrel{p}{\rightarrow}\Xi_{11};\,\,\tilde{\Xi}_{12}\stackrel{p}{\rightarrow}\Xi_{12};\,\,\tilde{\Xi}_{22}-\Xi_{22}\stackrel{p}{\rightarrow}0,
\]
as $\mu_{m_{W}-1}\rightarrow\infty$. 

It then follows that the diagonal elements of $\tilde{A}_{22}=\tilde{\Xi}_{12}^{\prime}\tilde{\Xi}_{12}+\tilde{\Xi}_{22}^{\prime}\tilde{\Xi}_{22}$
approach infinity as $\mu_{m_{W}-1}\rightarrow\infty$. Therefore,

\[
\tilde{D}_{11}\stackrel{p}{\rightarrow}I_{2};\,\,\tilde{D}_{22}\stackrel{p}{\rightarrow}I_{m_{W}-1},
\]
and
\[
\tilde{C}_{11}\stackrel{p}{\rightarrow}\Xi_{11}^{\prime}\Xi_{11};\,\,\tilde{C}_{12}\stackrel{p}{\rightarrow}0;\,\,\tilde{C}_{22}-\Xi_{22}^{\prime}\Xi_{22}\stackrel{p}{\rightarrow}0,
\]
as $\mu_{m_{W}-1}\rightarrow\infty$.

As $\tilde{D}'\tilde{D}=\tilde{D}\tilde{D}'=I_{m_{W}+1}$, we have
for the eigenvalues
\[
\hat{\kappa}_{j}\left(\tilde{C}\right)=\hat{\kappa}_{j}\left(\tilde{D}'\tilde{\Xi}'\tilde{\Xi}\tilde{D}\right)=\hat{\kappa}_{j}\left(\tilde{\Xi}'\tilde{\Xi}\right)=\hat{\kappa}_{j}\left(\Xi'\Xi\right),
\]
for $j=1,\ldots,m_{W}+1$. As $\tilde{C}$ approaches a block diagonal
matrix as $\mu_{m_{W}-1}\rightarrow\infty$, we then get that, using
generic notation, 
\begin{align*}
\left\{ \hat{\kappa}_{m_{W}}\left(\Xi'\Xi\right),\hat{\kappa}_{m_{W}+1}\left(\Xi'\Xi\right)\right\}  & =\text{\ensuremath{\left\{ \hat{\kappa}_{1}\left(\Xi_{11}^{\prime}\Xi_{11}\right),\hat{\kappa}_{2}\left(\Xi_{11}^{\prime}\Xi_{11}\right)\right\} ;}}\\
\hat{\kappa}_{j}\left(\Xi'\Xi\right) & \rightarrow\infty\text{, for }j=1,\ldots,m_{W}-1,
\end{align*}
as $\mu_{m_{W}-1}\rightarrow\infty$.

\subsection{Conditional Distribution of \texorpdfstring{$\hat{\kappa}_{p}|\hat{\kappa}_{p-1}$}{kappa-hat(p) | kappa-hat(p-1)}}
\label{subsec:Eigenvalues}

Denoting $\Gamma_{m}(a)=\pi^{m\left(m-1\right)/4}\prod_{i=1}^{m}\Gamma\left(a-\left(i-1\right)/2\right)$,
we can define the joint pdf for the (real) ordered eigenvalues $\hat{\kappa}_{1}\geq\hat{\kappa}_{2}\geq\cdot\cdot\cdot\geq\hat{\kappa}_{p}\geq0$
of the real non-central Wishart matrix $\Xi'\Xi\sim\mathcal{W}_{p}(k,I_{p},\mathcal{M}'\mathcal{M})$
as (\citealp{james1964distributions}) 
\[
f(x_{1},...,x_{p})={}_{0}F_{1}\Big(\frac{1}{2}k;\frac{1}{4}\Omega,S\Big)K\prod_{i=1}^{p}e^{-\kappa_{i}/2}e^{-x_{i}/2}x_{i}^{\alpha}\prod_{i<j}^{p}\left(x_{i}-x_{j}\right),
\]
where $x_{1}\geq x_{2}\geq\cdot\cdot\cdot\geq x_{p}\geq0$, $\alpha=\left(k-p-1\right)/2$,
$\Omega=\text{diag}\left(\kappa_{1},\kappa_{2},...,\kappa_{p}\right)$,
$S=\text{diag}\left(x_{1},x_{2},...,x_{p}\right)$,
and $K=\frac{\pi^{(p^{2}/2)}}{2^{kp/2}\Gamma_{p}(k/2)\Gamma_{p}(p/2)}$.
As before we have that $\kappa_{1}\geq\kappa_{2}\geq\cdot\cdot\cdot\geq\kappa_{p}$
are the eigenvalues of $\mathcal{M}'\mathcal{M}$. Moreover, the hypergeometric
function $_{0}F_{1}$ of the matrix arguments $\frac{1}{4}\Omega$
and $S$ is defined as 
\[
_{0}F_{1}\Big(\frac{1}{2}k;\frac{1}{4}\Omega,S\Big)=\sum_{n=0}^{\infty}\sum_{\eta\in\mathcal{P}_{n}}\frac{C_{\eta}\left(\frac{1}{4}\Omega\right)C_{\eta}\left(S\right)}{\left(\frac{1}{2}k\right)_{\eta}C_{\eta}\left(I_{p}\right)n!},
\]
where $\mathcal{P}_{n}$ is the set of all partitions of $n$ and
$\left(\frac{1}{2}k\right){}_{\eta}$ is the generalized Pochhammer
symbol defined as 
\[
\left(\frac{1}{2}k\right)_{\eta}:=\left(\frac{1}{2}k\right)_{\left(\eta_{1},...,\eta_{p}\right)}:=\prod_{i=1}^{p}\left(\frac{1}{2}k-\frac{i-1}{2}\right)_{\eta_{i}}.
\]
Every $\eta\in\mathcal{P}_{n}$ is defined to be a tuple $\eta=\left(\eta_{1},...,\eta_{p}\right)$
such that $\eta_{1}\geq\eta_{2}\geq\cdot\cdot\cdot\geq\eta_{p}\geq1$
and $\eta_{1}+\cdot\cdot\cdot+\eta_{p}=n$. Finally, we define $C_{\eta}\left(\frac{1}{4}\Omega\right)$
and $C_{\eta}\left(S\right)$ as the zonal polynomials of $\frac{1}{4}\Omega$
and $S$, satisfying $\sum_{\eta\in\mathcal{P}_{n}}C_{\eta}\left(S\right)=\left(tr\left(S\right)\right)^{n}=\left(x_{1}+...+x_{p}\right)^{n}$
and $\sum_{\eta\in\mathcal{P}_{n}}C_{\eta}\left(\frac{1}{4}\Omega\right)=\left(\frac{1}{4}tr\left(\Omega\right)\right)^{n}=\frac{1}{4^{n}}\left(\kappa_{1}+...+\kappa_{p}\right)^{n}$,
respectively. We use the package Zonal.Sage in SageMath (\citealp{jiu2020calculation})
to calculate these zonal polynomials. \citet{jiu2020calculation}
further elaborate on how these zonal polynomials can be calculated,
following the results from \citet{muirhead2009aspects} and \citet{james1964distributions}.

This joint pdf can be used to calculate the conditional pdf. For $p=3$
this results in 
\[
f_{\hat{\kappa}_{3}|\hat{\kappa}_{2}}\left(x_{3}|x_{2}\right)=\frac{\int_{x_{2}}^{\infty}f\left(x_{1},x_{2},x_{3}\right)dx_{1}}{\int_{x_{2}}^{\infty}\int_{0}^{x_{2}}f\left(x_{1},x_{2},x_{3}\right)dx_{3}dx_{1}}.
\]
From this, we can calculate the conditional cdf as usual. Table \ref{tab:CondCDF}
and Figures \ref{fig:Ccdf} and \ref{fig:Ccdfapp3} highlight
for the $p=3$ case that we have
\[
P\left(\hat{\kappa_3}\leq a|\hat{\kappa}_{2}=x\right)_{\kappa_{1}<\infty}>P\left(\hat{\kappa}_{3}\leq a|\hat{\kappa}_{2}=x\right)_{\kappa_{1}=\infty}
\]
\[
P\left(\hat{\kappa_3}>a|\hat{\kappa}_{2}=x\right)_{\kappa_{1}<\infty}\leq P\left(\hat{\kappa}_{3}>a|\hat{\kappa}_{2}=x\right)_{\kappa_{1}=\infty}.
\]
Then, combining this with the findings from GKM
it shows that, under the null $H_{0}:\beta=\beta_{0}$, 
\[
\mathbb{P}\left(AR_{n}\left(\beta_{0}\right)>c_{1-\alpha}\left(\hat{\kappa}_{2},k-m_{W}\right)\right)_{\kappa_{1}<\infty}\leq\mathbb{P}\left(AR_{n}\left(\beta_{0}\right)>c_{1-\alpha}\left(\hat{\kappa}_{2},k-m_{W}\right)\right)_{\kappa_{1}=\infty}\leq\alpha.
\]
As we have the joint pdf for the eigenvalues of the non-central real
Wishart matrix, this can be extended to any choice of $p$, indicating
that our test has correct size for any $m_{W}$.

We illustrate this further in the figures below, where we simulate
the conditional cdfs as described in the main text. These results
confirm the monotonicity for different settings of $p$ and $k$.

\begin{figure}[ht]
\begin{centering}
\noindent
\includegraphics[width=0.49\linewidth]{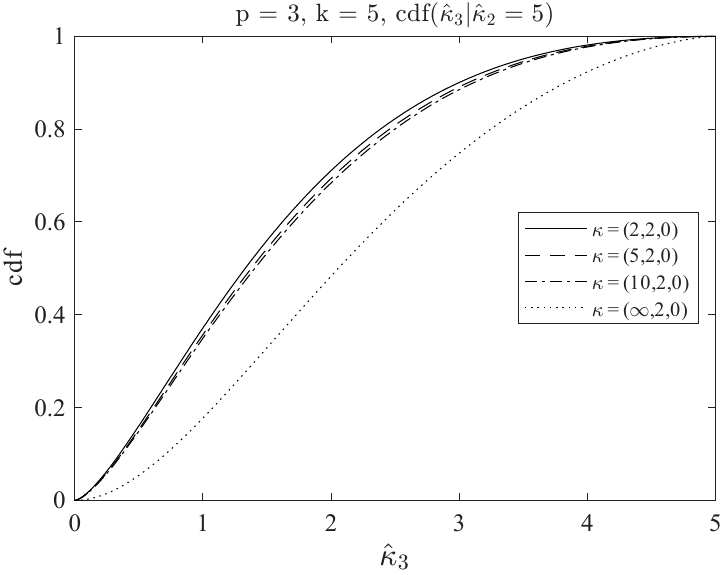}\hfill%
\includegraphics[width=0.49\linewidth]{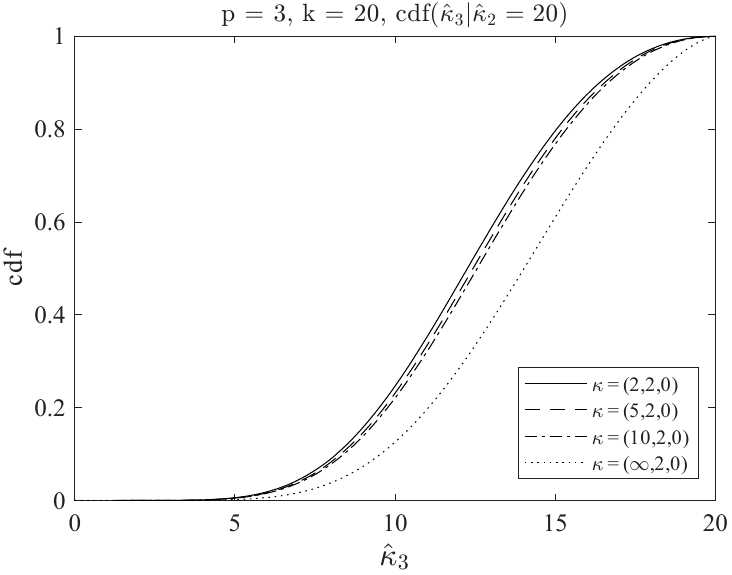}
\par\end{centering}
\caption{\label{fig:Ccdfapp3} Conditional cdfs, $p=3$}
\end{figure}

\begin{figure}[ht]
\begin{centering}
\noindent
\includegraphics[width=0.49\linewidth]{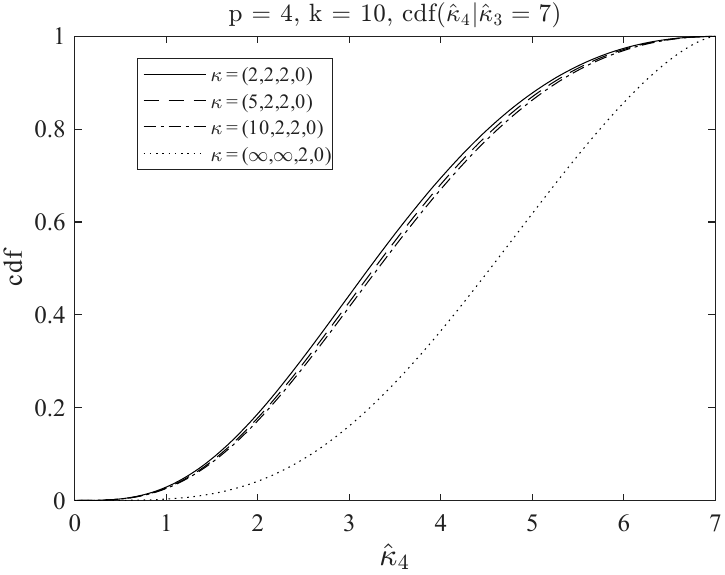}\hfill%
\includegraphics[width=0.49\linewidth]{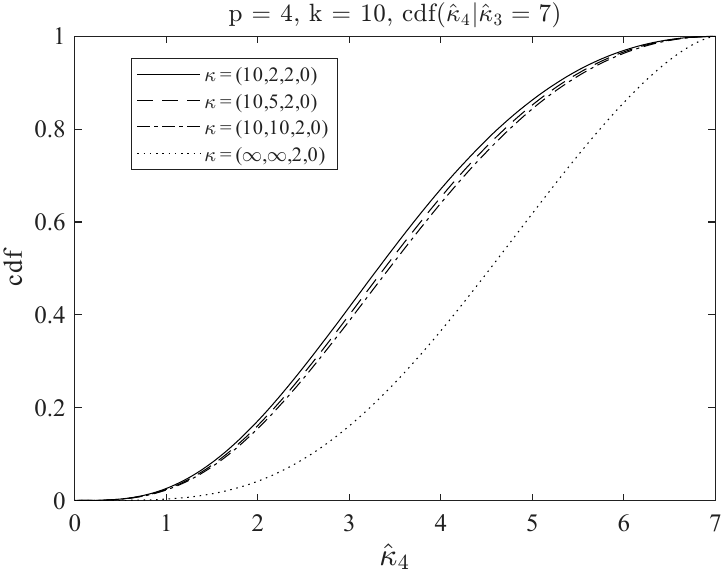}
\par\end{centering}
\caption{\label{fig:Ccdfapp4} Conditional cdfs, $p=4$}
\end{figure}

\begin{figure}[ht]
\begin{centering}
\noindent
\includegraphics[width=0.49\linewidth]{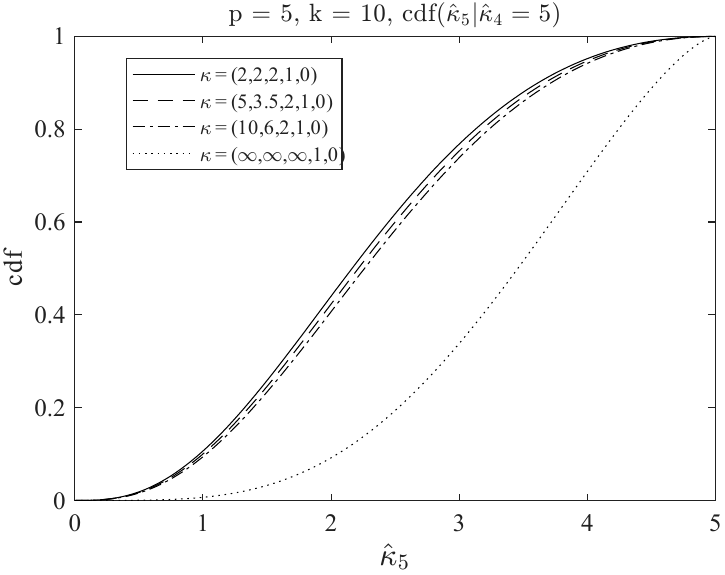}\hfill%
\includegraphics[width=0.49\linewidth]{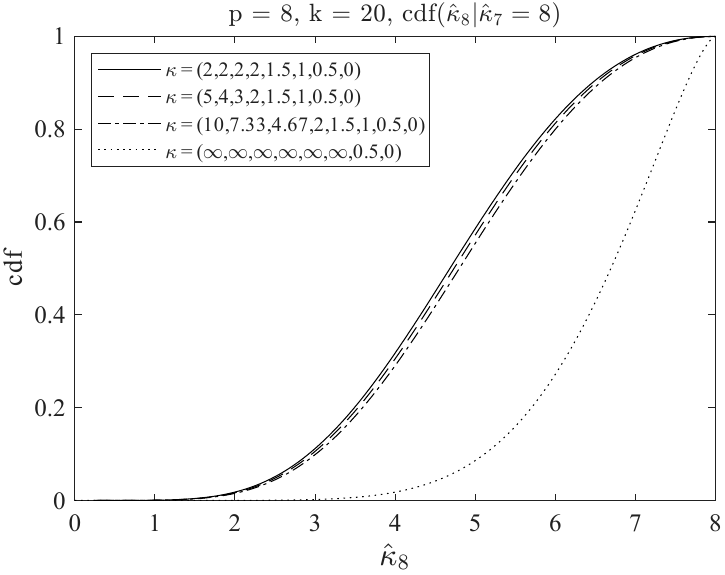}
\par\end{centering}
\caption{\label{fig:Ccdfapp5&8} Conditional cdfs, $p=5$ and $p=8$}
\end{figure}

\subsection{Monte Carlo Parameter Values}

\label{subsec:Monte-Carlo-Parameter}

The specific values for $\Sigma$, $\pi_{x}$ and $\Pi_{W}$ for the analysis in Section \ref{sec:Sims} are given by
\[
\Sigma=\left(\begin{array}{ccccc}
1 & 0.1 & 0.3 & 0.2 & 0.8\\
0.1 & 1 & 0.3 & 0.2 & 0.1\\
0.3 & 0.3 & 1 & 0.3 & 0.2\\
0.2 & 0.2 & 0.3 & 1 & 0.3\\
0.8 & 0.1 & 0.2 & 0.3 & 1
\end{array}\right);\pi_{x}=\left(p_x/\sqrt{kn}\right)\left(\begin{array}{c}
1\\
1\\
1\\
-1\\
1\\
1\\
1
\end{array}\right),
\]
with $p_x=4$ in Figure \ref{fig:Power} and $p_x=8$ in Figure
\[
\Pi_{W}=A\left(\begin{array}{ccc}
\sqrt{\kappa_{1}} & 0 & 0\\
0 & \sqrt{\kappa_{2}} & 0\\
0 & 0 & \sqrt{\kappa_{3}}
\end{array}\right)\Sigma_{V_{W}V_{W}.\varepsilon}^{1/2},
\]
where
\[
A=\frac{1}{\sqrt{n}}\left(\begin{array}{ccc}
\frac{1}{\sqrt{3}} & 0 & 0\\
\frac{1}{\sqrt{3}} & 0 & 0\\
\frac{1}{\sqrt{3}} & 0 & 0\\
0 & \frac{1}{\sqrt{2}} & 0\\
0 & \frac{1}{\sqrt{2}} & 0\\
0 & 0 & \frac{1}{\sqrt{2}}\\
0 & 0 & \frac{1}{\sqrt{2}}
\end{array}\right),
\]
so that
\[
\mathbb{E}_{Z}\left(\Theta_{W}^{\prime}\Theta_{W}\right)=n\Sigma_{V_{W}V_{W}.\varepsilon}^{-1/2}\Pi_{W}^{\prime}\Pi_{W}\Sigma_{V_{W}V_{W}.\varepsilon}^{-1/2}=\left(\begin{array}{ccc}
\kappa_{1} & 0 & 0\\
0 & \kappa_{2} & 0\\
0 & 0 & \kappa_{3}
\end{array}\right).
\]

Figure \ref{fig:PowerX} shows the power graphs for the tests when setting $p_x=8$ in the formulation of $\pi_x$ above.
\begin{figure}
\begin{centering}
\includegraphics[scale=0.65]{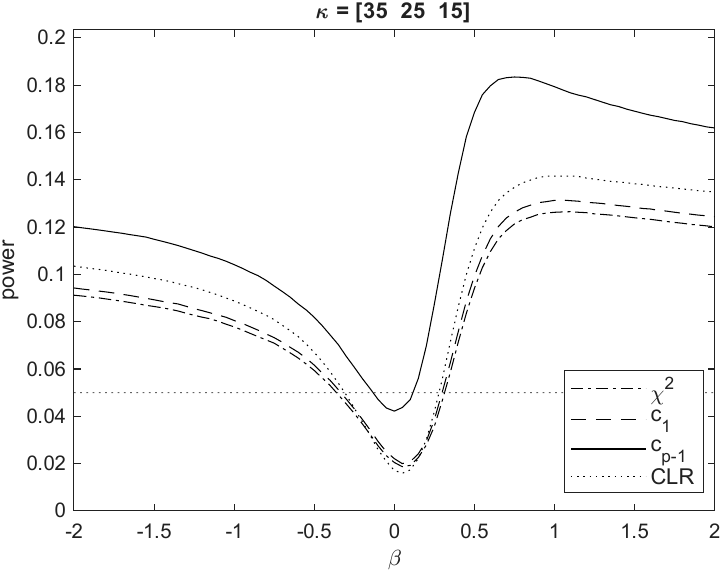}\includegraphics[scale=0.65]{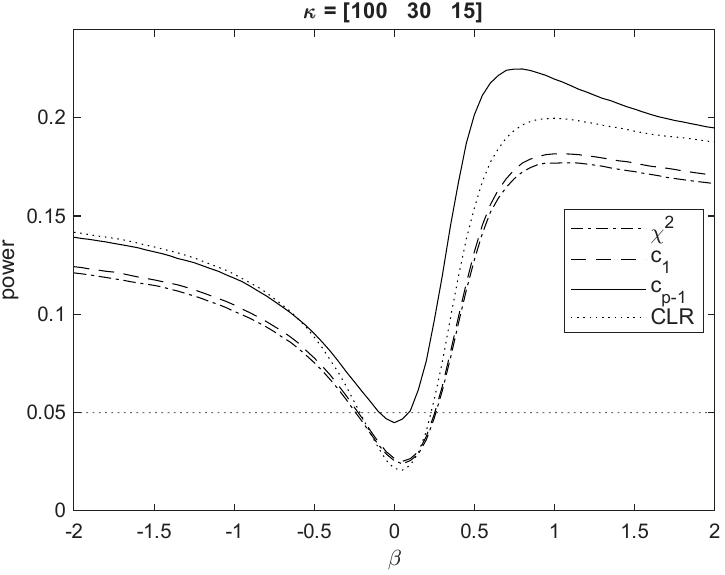}
\par\end{centering}
\begin{centering}
\includegraphics[scale=0.65]{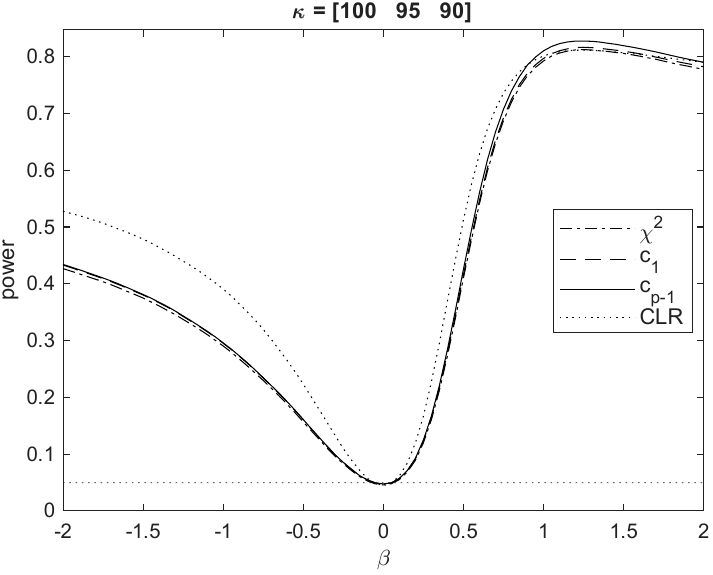}
\par\end{centering}
\caption{\label{fig:PowerX} Power of $\phi_{\chi^{2}}$, $\phi_{c_1}$, $\phi_{c_{p-1}}$ and CLR test of \cite{Kleibergen2021}
for testing $H_{0}:\beta=0$ at level $\alpha=0.05$. $m_{W}=3$,
$k=7$, $n=250$, $\mu_x^2=64$, for different values of $\kappa=\left(\kappa_{1},\kappa_{2},\kappa_{3}\right)$.}
\end{figure}

Figure \ref{fig:PowerT} shows the power graphs for the tests in the same setting as for Figure \ref{fig:Power}, but instead of the multivariate normal distribution of the errors, their distribution is multivariate $t$, with 4 degrees of freedom and the same variance matrix $\Sigma$ as above. The graphs are virtually identical to the ones in Figure \ref{fig:Power}.
\begin{figure}
\begin{centering}
\includegraphics[scale=0.65]{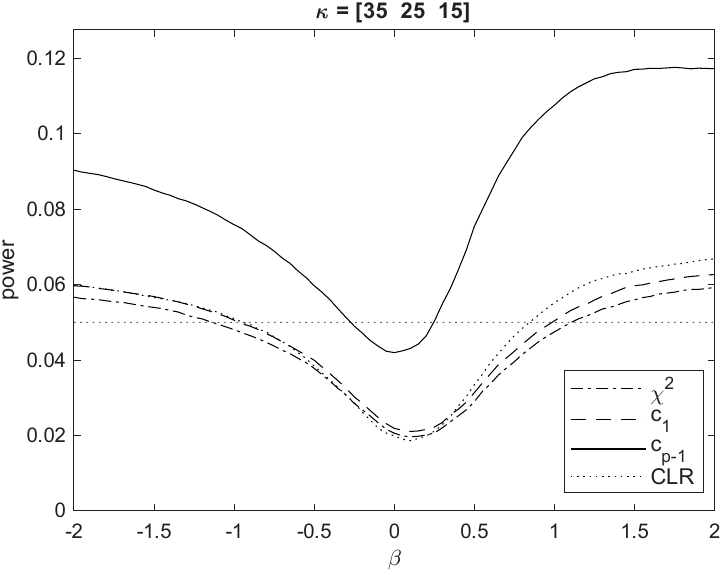}\includegraphics[scale=0.65]{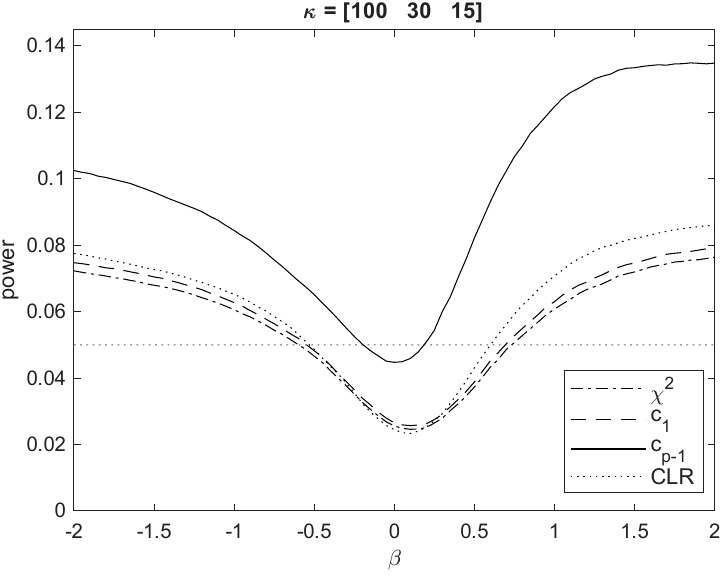}
\par\end{centering}
\begin{centering}
\includegraphics[scale=0.65]{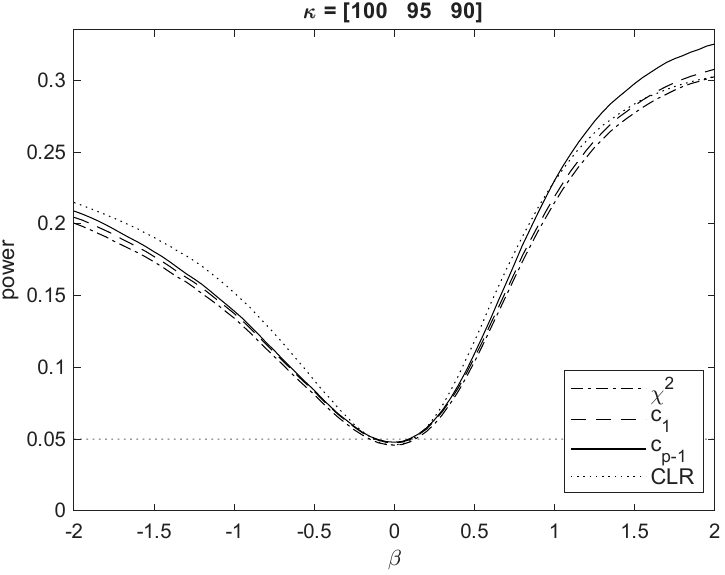}
\par\end{centering}
\caption{\label{fig:PowerT} Power of $\phi_{\chi^{2}}$, $\phi_{c_1}$, $\phi_{c_{p-1}}$ and CLR test of \cite{Kleibergen2021}
for testing $H_{0}:\beta=0$ at level $\alpha=0.05$. $m_{W}=3$,
$k=7$, $n=250$, $\mu_x^2=16$, for different values of $\kappa=\left(\kappa_{1},\kappa_{2},\kappa_{3}\right)$. The errors are mulitvariate $t_4$ distributed.}
\end{figure}

\subsection{Conditional Heteroskedasticity with Kronecker Product Structure}
\label{sec:Het}

\cite*{Guggenberger_Kleibergen_Mavroeidis_2024} (GKM24) showed that their conditional critical value function, conditioning on the largest eigenvalue of a matrix defined below, also applies to a subvector $AR$ test statistic that is robust to a Kronecker product form of conditional heteroskedasticity. From their results, it follows again that our approach, conditioning on the second-smallest eigenvalue of this matrix also maintains size control and has strictly higher power.

As in Section \ref{sec:Sims}, GKM24 assumes that the random vectors $\left(\varepsilon_i,Z_i',V_{Xi}',V_{Wi}'\right)$ for $i=1,\ldots,n$ are i.i.d.\ with distribution $F$. Again, let $U_i=\left(\varepsilon_i+V_{Wi}'\gamma,V_{Wi}'\right)'$. Then GKM24 considers the parameter space $\mathcal{F}_{AKP,a_n}$ for $\left(\gamma,\Pi_W,\Pi_X,F\right)$ under the null hypothesis $H_0:\beta=\beta_0$ that is larger than the homoskedastic $\mathcal{F}$ in (\ref{eq:parsF}), imposing an approximate Kronecker product (AKP) structure for the covariance matrix, given by
\begin{align}
\label{eq:F_kron}
    \mathcal{F}_{AKP,a_n} &=\{(\gamma, \Pi_W, \Pi_X, F): \gamma\in \mathbb{R}^{m_W}, \Pi_W\in\mathbb{R}^{k \times m_W}, \Pi_X \in \mathbb{R}^{k\times m_X}, \notag  \\
    &\mathbb{E}_F[||T_i||^{2+\delta_1}]\leq B, \text{ for } T_i \in \{\text{vec}(Z_iU_i'), ||Z_i||^2\}, \notag \\
    &\mathbb{E}_F[Z_iV_i']=0^{k \times (m+1)}, \mathbb{E}_F\left(\text{vec}(Z_iU_i')(\text{vec}(Z_iU_i'))'\right)=G_F\otimes H_F+\Upsilon_n, \notag \\
    &\kappa_{min}(A)\geq\delta_2 \text{ for } A \in \{ \mathbb{E}_F[Z_i'Z_i], G_F, H_F \}\} 
\end{align}
for symmetric matrices $\Upsilon_n\in\mathbb{R}^{kp\times kp}$ such that $||\Upsilon_n||\leq a_n$, positive definite symmetric matrices $G_F \in \mathbb{R}^{p\times p}$ (whose upper left element is normalized to 1) and $H_F\in \mathbb{R}^{k\times k}, \delta_1, \delta_2>0, B< \infty$.

Let $\hat{R}_n \coloneqq n^{-1}\sum^n_{i=1}f_if_i'$, where $f_i:=((M_Zy_0)_i, (M_ZW)_i')'$, then GKM24 specify estimators $\hat{G}_n$ and $\hat{H}_n$ for $G_F$ and $H_F$ respectively, given by
\begin{equation*}
\label{eq:minimizers_text}
    (\hat{G}_n, \hat{H}_n)=\arg\min||G \otimes H - \hat{R}_n||,
\end{equation*}
where the minimum is taken over ($G,H$) for $G \in \mathbb{R}^{p \times p}, H\in \mathbb{R}^{k\times k}$ being positive definite, symmetric matrices, and normalized such that the upper left element of $G$ equals 1. GKM24 (pp 963-964) provide details how to obtain these estimators by means of a singular value decomposition.

The feasible subvector $AR$ statistic when the variance matrix has an approximate Kronecker structure is then defined as\footnote{For this $AR$ statistic to be invariant to nonsingular transformations of the instruments matrix, the instruments need to be standardised such that $Z'Z/n=I_p$ before calculating the test statistic.}

\begin{align*}
AR_{AKP,n}\left(\beta_{0}\right) & \coloneqq \min_{\tilde{\gamma}\in\mathbb{R}^{m_{W}}} n^{-1}\frac{\left(y_{0}-W\tilde{\gamma}\right)^{\prime}Z\hat{H}_n^{-1}Z'\left(y_{0}-W\tilde{\gamma}\right)}{\left(1,-\tilde{\gamma}^{\prime}\right)\hat{G}_n\left(1,-\tilde{\gamma}^{\prime}\right)^{\prime}}\\
 & =\min\,\text{eval}\left(n^{-1}\left(y_{0},W\right)^{\prime}Z\hat{H}_n^{-1}Z'\left(y_{0},W\right)\hat{G}_n^{-1}\right).
\end{align*}

Denote the $p$ eigenvalues of the matrix $n^{-1}\left(y_{0},W\right)^{\prime}Z\hat{H}_n^{-1}Z'\left(y_{0},W\right)\hat{G}_n^{-1}$
by $\hat{\kappa}_{1n}\geq\hat{\kappa}_{2n}\geq\ldots\geq\hat{\kappa}_{pn}$,
so that $AR_{AKP,n}\left(\beta_{0}\right)=\hat{\kappa}_{pn}$. GKM24 (Theorem 1, p 965) show that for the parameter space $\mathcal{F}_{AKP,a_n}$ defined in (\ref{eq:F_kron}) the conditional subvector $AR_{AKP}$ test rejects $H_0$ at nominal size $\alpha$ asymptotically if
$$
AR_{AKP,n}\left(\beta_0\right)>c_{1-\alpha}\left(\hat\kappa_{1n},k-m_W\right),
$$
where $c_{1-\alpha}(.,.)$ is the same conditional critical value function as in Sections 2 and 3.

As GKM24 (p 994) state, the same proof as for Theorem 5 in GKM applies to Theorem 1 in GKM24. As discussed in Section \ref{sec:Sims}, the limiting null rejection probabilities are therefore equal to the finite sample ones as derived in Section \ref{sec:finite_sample_analysis}. Because of this convergence, it follows from the results in Section \ref{sec:finite_sample_analysis} that the conditional subvector $AR_{AKP}$ test also rejects $H_0$ at nominal size $\alpha$ asymptotically if we condition on the second-smallest eigenvalue, that is if\footnote{A formal proof is obtained by considering $0\leq q<m_W$ in Proposition 5 of GKM24 (p 993), see for further details \cite{hoekstra2026}.}
\[
AR_{AKP,n}\left(\beta_0\right)>c_{1-\alpha}\left(\hat\kappa_{(p-1)n},k-m_W\right).
\]

\subsection{Power Bound}
\label{sec:PowerBound}
We compute the power bound for all tests that depend on the statistic
$(\hat{\kappa}_1,\hat{\kappa}_2,\hat{\kappa}_3)$. These are based on the least favorable distribution for the nuisance parameters $(\kappa_1,\kappa_2)$ under the null that $\kappa_3=0$.

We consider the least favorable distribution (ALFD) method of \cite{elliott2015nearly}, following what GKM reported for
$m_W=1$ (Section 2.4).

Recall that under $H_0:\beta=\beta_0$, the joint distribution of $(\hat{\kappa}_1,\ldots,\hat{\kappa}_p)$ only depends on the vector of eigenvalues $(\kappa_1,\ldots,\kappa_{m_W})$ of $\Theta_W'\Theta_W$,
where $\Theta_W\in\mathbb{R}^{k\times m_W}$ appears in the noncentrality matrix $\mathcal{M}=(0,\Theta_W)$ of $\Xi\sim \mathcal{N}(\mathcal{M},I_{kp})$. From the definition of $\Theta_W$, we can see that if $\Pi_W$ ranges through all matrices in $\mathbb{R}^{k\times m_W}$, then $(\kappa_1,...,\kappa_{m_W})'$ ranges through all vectors in $[0,\infty)^{m_W}$.

Define $A:=\mathbb{E}\left[Z'(y-X\beta_0,W)\right]\in\mathbb{R}^{k\times p}$
and consider the null hypothesis
\[
H_0':\rho(A)\leq m_W
\qquad \text{against} \qquad
H_1':\rho(A)=p,
\]
where $\rho(A)$ denotes the rank of matrix $A$. Clearly, whenever $H_0$ holds, $H_0'$ holds too, but the reverse is not true; as explained in GKM (Appendix A), $H_0'$ holds iff $\Pi_W$ is rank deficient or $\Pi_X(\beta-\beta_0)\in\operatorname{span}(\Pi_W)$. They further show that under $H_0'$, the joint distribution of $(\hat{\kappa}_1,\ldots,\hat{\kappa}_p)$ is the same as the one of the vector of eigenvalues of a Wishart matrix $\mathcal{W}_p(k,I_p,\mathcal{M}'\mathcal{M})$ with rank deficient noncentrality matrix and, therefore, depends only on the vector of the largest $m_W$ eigenvalues $(\kappa_1,\ldots,\kappa_{m_W})'\in\mathbb{R}^{m_W}$ of $\mathcal{M}'\mathcal{M}$.

The important implication of the result is that any test $\varphi(\hat{\kappa}_1,\ldots,\hat{\kappa}_p)\in[0,1]$ for some measurable function $\varphi$ that has size bounded by $\alpha$ under $H_0$ also has size (in the parameters $(\beta,\gamma,\Pi_X,\Pi_W)$) bounded by $\alpha$ under
$H_0'$. In particular, no test $\varphi(\hat{\kappa}_1,\ldots,\hat{\kappa}_p)$ that controls size under $H_0$ has power exceeding size under alternatives $H_0'\backslash H_0$.

The testing problem we focus on is for $m_W=2$,
\[
H_0:\kappa_3=0,\quad \kappa_1\geq\kappa_2\geq0
\qquad \text{against} \qquad
H'_1:\kappa_3>0,\quad \kappa_1\geq\kappa_2\geq\kappa_3,
\]
where $\kappa_1,\kappa_2,\kappa_3$ are the eigenvalues $\kappa_i(\mathcal{M}'\mathcal{M})$ of the $3\times3$ noncentral Wishart matrix $\Xi'\Xi\sim \mathcal{W}_3(k,I_3,\mathcal{M}'\mathcal{M})$ and $\hat{\kappa}_i(\Xi'\Xi)$. The joint density of the eigenvalues $f_{\hat{\kappa}_1,\ldots,\hat{\kappa}_3}(\cdot)$ is given in Appendix \ref{subsec:Eigenvalues}.

All simulations in this section are performed using importance sampling. The parameter space of $\eta=(\kappa_1,\kappa_2)$ is discretized into $N_\eta=65$. This grid combines nine nuisance-shape directions, seven logarithmically spaced strength levels between 0.1 and 100, the origin, and the configured point. The calibrated EMW upper bound $\bar{\pi}$ remains valid for any chosen null mixture; the grid affects tightness, not upper-bound validity. Following GKM, the reported curve uses the grid-adjusted $\widetilde{\pi}$. As we compare our test's power curves to this upper bound, we strive for a tight upper bound, but in the case where there is proximity of our power curve to a
looser variant of this power curve does not invalidate our conclusion, and could only make them stronger. This justifies this heuristic grid size choice, similarly as the $N_{\kappa_1}$ in GKM. 

We will compute point-optimal (PO) power bounds over a grid of point alternatives. Let $F$ denote a distribution over $H_1$, so that $H_{1,F}\in H_1$ is a point alternative, and let $g$ denote the density of the data under $H_{1,F}$. For the power envelope, we consider one-point distributions $F$. We calculate the power bound for $(\kappa_1,\kappa_2) \in \{(100,15), (50,15), (15,15)\}$, as in this $m_W=2$ extension of GKM, more points would make the calculations overwhelming. We calculate the pointwise EMW envelope at 17 points along the model-implied $\beta$-path. The respective $\beta_r$, with DGP allows us to calculate the $(\kappa_{1r},\kappa_{2r})$ and $ g_r= f_{\hat{\kappa}_1,\hat{\kappa}_2,\hat{\kappa}_3}
(\cdot;\kappa_{1r},\kappa_{2r},\kappa_{3r})$.

Moreover, the null discretization is slightly more involved than in the $m_W=1$ case. The nuisance space is now $K_0=\{(\nu_1,\nu_2):\nu_1\geq\nu_2\geq0\}$.
A nonzero point in this cone can be represented as $(\nu_1,\nu_2)=s(1,d_2), s>0, 1\geq d_2\geq0$, separating overall strength $s$ from nuisance shape $d_2$.

For the $(\kappa_1,\kappa_2)\in \{(100,15), (50,15), (15,15)\}$ design, we use a common grid
containing $H=1+9\times7 + 1=65$ null points. It consists of the origin, nine nuisance-shape directions, seven strength levels along each direction $s\in\{10^{-1},10^{-1/2},1,10^{1/2},10,10^{3/2},100\}$,
and the configured point. The shape directions include the configured ratio $\kappa_2/\kappa_1$, the rank-one direction, the equal-eigenvalue direction, and the smallest and largest ratios of the second-to-first noncentrality eigenvalue along a fixed $\beta$-path. The construction of this grid does not depend on the number of $\beta$ values subsequently plotted, and within each design, the same ordered grid is reused at every alternative.

That said, the rest of the implementation is a straightforward extension to $m_W=2$ of what is reported in GKM. Hence, for this, we refer to their paper. 

\medskip

\noindent\textbf{Acknowledgements}

\noindent We would like to thank an associate editor, two anonymous referees and Sophocles Mavroeidis for very helpful comments and suggestions, and Daniel Gagliardi and Jakub Grohmann for their research assistance early on in the project.

\bibliographystyle{ecta}
\bibliography{SubAR}

\end{document}